%% file: main.tex
\documentclass[11pt, letterpaper]{article}
\usepackage{amsmath, amssymb, amsthm, amsfonts, bbm, commath}
\usepackage[inline]{enumitem}
\usepackage{tikz}
\usetikzlibrary{arrows.meta,shapes}
\usepackage{subcaption}
\usepackage{authblk}
\usepackage{graphicx}
\usepackage{hyperref}
\usepackage{xcolor}
\usepackage[ruled]{algorithm2e}
\usepackage{cancel}
\usepackage{multirow}
\usepackage{calc}

\usepackage[letterpaper, margin=1in]{geometry}

\setlist{topsep=0pt,labelindent=0pt,labelsep=*,leftmargin=*,itemindent=0pt,itemsep=0pt}

\title{
	Exact Byzantine Consensus on Undirected Graphs under {\em Local Broadcast} Model
	\thanks{
	    This research is supported in part by the National Science Foundation awards 1409416 and 1733872, and Toyota InfoTechnology Center. Any opinions, findings, and conclusions or recommendations expressed here are those of the authors and do not necessarily reflect the views of the funding agencies or the U.S. government.
    }
}

\author[1]{
	Muhammad Samir Khan
}

\author[1]{
    Syed Shalan Naqvi
}

\author[2]{
    Nitin H. Vaidya
}

\affil[1]{
	Department of Computer Science \protect\linebreak
	University of Illinois at Urbana-Champaign \protect\linebreak
	\texttt{\{mskhan6,naqvi5\}@illinois.edu} \protect\linebreak
}

\affil[2]{
	Department of Computer Science \protect\linebreak
	Georgetown University \protect\linebreak
	\texttt{nv198@georgetown.edu}
}

\include{macros}
\begin{document}
	\maketitle

    \begin{abstract}
		\normalsize
		\input{abstract.tex}
	\end{abstract}

	\section{Introduction} \label{section introduction}
		\input{introduction.tex}

	\section{Related Work} \label{section related work}
		\input{relatedWork.tex}

	\section{System Model and Notation} \label{section model}
		\input{model.tex}

	\section{Necessary Conditions under Local Broadcast} \label{section necessity}
		\input{necessity.tex}

%
%

	\section{Byzantine Consensus Algorithm under Local Broadcast} \label{section sufficiency}
		\input{sufficiency.tex}

		\subsection{Proposed Algorithm} \label{section algorithm}
			\input{algorithm.tex}

		\subsection{Proof of Correctness of Algorithm \ref{algorithm consensus}} \label{section correctness}
			\input{correctness.tex}
			
		\subsection{An Efficient Algorithm} \label{section efficienct}
		    \input{efficient.tex}
		
	\section{Hybrid Model} \label{section hybrid}
		\input{hybridFaults.tex}

	\section{Summary} \label{section discussion}
		\input{discussion.tex}

	\bibliographystyle{abbrv}
	\bibliography{bib}

	\appendix
	\section{Proofs of Section \ref{section necessity}} \label{section proofs necessity}
		\input{proofsNecessity.tex}
	
	\section{Proofs of Section \ref{section sufficiency}} \label{section proofs sufficiency}
		\input{proofsSufficiency.tex}
	
	\section{Efficient Algorithm for $2f$-Connected Graphs} \label{section efficient algorithm}
	
		\input{efficientAlgorithm.tex}
		
		\input{proofsEfficient.tex}
	
	    
	\section{Proofs of Section \ref{section hybrid}} \label{section proofs hybrid}
	    This appendix pertains to the hybrid model.
		\subsection{Necessity of the conditions in Theorem \ref{theorem hybrid}} \label{section proofs hybrid necessity}
			\input{proofsHybridNecessity.tex}
		\subsection{Sufficiency of the conditions in Theorem \ref{theorem hybrid}} \label{section proofs hybrid sufficiency}
			\input{proofsHybridSufficiency.tex}

\end{document}

%% file: macros.tex
\newtheorem{theorem}{Theorem}[section]
\newtheorem{lemma}[theorem]{Lemma}

\newtheorem{definition}[theorem]{Definition}

\newtheorem{observation}[theorem]{Observation}

\renewenvironment{proof}{\noindent{\bf Proof:} \hspace*{1mm}}{
	\hspace*{\fill} $\Box$ }
\newenvironment{proof_of}[1]{\noindent {\bf Proof of #1:}
	\hspace*{1mm}}{\hspace*{\fill} $\Box$ }

\newcommand{\floor}[1]{\left\lfloor {#1} \right\rfloor}
\newcommand{\ceil}[1]{\left\lceil #1 \right\rceil}

%% file: abstract.tex
%

This paper considers the Byzantine consensus problem for nodes with binary inputs. The nodes are interconnected by a network represented as an undirected graph, and the system is assumed to be synchronous.
Under the classical point-to-point communication model, it is well-known \cite{DOLEV198214} that the following two conditions are both necessary and sufficient to achieve Byzantine consensus among $n$ nodes in the presence of up to $f$ Byzantine faulty nodes: $n \ge 3f+1$ and vertex connectivity at least $2f+1$. In the classical point-to-point communication model, it is possible 
for a faulty node to \emph{equivocate}, i.e., transmit conflicting information to different neighbors. Such equivocation is possible because messages sent by a node to one of its neighbors are {\em not overheard} by other neighbors.

This paper considers the {\em local broadcast} model. In contrast to the point-to-point communication model, in the local broadcast  model, messages sent by
a node
are received identically by {\em all of its neighbors}.
Thus, under the local broadcast model, attempts by a node to send conflicting information
can be detected by
its
neighbors. Under this model, we show that the following two conditions are both necessary and sufficient for Byzantine consensus: vertex connectivity at least $\lfloor 3f/2 \rfloor + 1$ and minimum node degree at least $2f$. Observe that the local broadcast model results in a lower requirement for connectivity and the number of nodes $n$, as compared to the point-to-point communication model.

We extend the above results to a {\em hybrid model} that allows some of the Byzantine faulty nodes to equivocate. The hybrid model bridges the gap between the point-to-point and local broadcast models, and helps to precisely characterize the trade-off between equivocation and network requirements.

%

%% file: introduction.tex
This paper considers Byzantine consensus for nodes with binary inputs. The nodes are interconnected by a communication network represented as an undirected graph, and the system is assumed to be synchronous.
Under the classical point-to-point communication model, it is well-known \cite{DOLEV198214} that the following two conditions are both necessary and sufficient to achieve Byzantine consensus among $n$ nodes in the presence of up to $f$ Byzantine faulty nodes: $n \ge 3f+1$ and vertex connectivity at least $2f+1$. In the classical point-to-point communication model, it is possible 
for a faulty node to \emph{equivocate} \cite{Chun:2007:AAM:1323293.1294280}, i.e., transmit conflicting information to different neighbors. Such equivocation is possible because messages sent by a node to one of its neighbors are {\em not overheard} by other neighbors.

In contrast, in the {\em local broadcast} model \cite{Bhandari:2005:RBR:1073814.1073841, Koo:2004:BRN:1011767.1011807} considered in this paper, a message sent by any node
is received identically by {\em all of its neighbors} in the communication network.
Thus, under the local broadcast model, attempts by a node to equivocate (i.e., send conflicting information to its neighbors) can be detected by
its
neighbors.

This paper obtains tight necessary and sufficient conditions on the underlying communication network to be able to achieve Byzantine consensus under the {\em local broadcast} model. As summarized in Section \ref{section related work}, although there has been significant work \cite{Amitanand:2003:DCP:872035.872065, Bhandari:2005:RBR:1073814.1073841, Clement:2012:PN:2332432.2332490, Considine2005, Fitzi:2000:PCG:335305.335363, Franklin2000, FranklinHypergraphsPrivacy2004, Jaffe:2012:PEB:2332432.2332491, Koo:2004:BRN:1011767.1011807, Koo:2006:RBR:1146381.1146420, Li7516067, Pagourtzis2017, Pelc:2005:BLB:1062032.1711012, Rabin:1989:VSS:73007.73014, Ravikant10.1007/978-3-540-30186-8_32, TSENG2015512, Wang2001} that uses either the local broadcast model or other models that restrict equivocation, tight necessary and sufficient conditions for Byzantine consensus under the local broadcast model have not been obtained previously.
In particular, this paper makes the following contributions, some of which have been documented elsewhere \cite{ArxivTightReport, NaqviMaster, NaqviBroadcast} as well:

\begin{enumerate}[label=\arabic*)]
	\item {\bf Necessary and sufficient condition:}
	In Sections \ref{section necessity} and \ref{section sufficiency},
    we establish that, under the {\em local broadcast} model, to achieve Byzantine consensus, it is necessary and sufficient for
	the communication graph to have vertex connectivity at least $\floor{3f/2}+1$ and minimum degree at least $2f$. Observe that the local broadcast model results in a lower connectivity requirement as compared to the ($2f+1$)-connectivity required under the point-to-point communication model.
	
	\item {\bf Efficient algorithm:} We constructively prove the sufficiency of the tight condition stated above by presenting an algorithm that achieves Byzantine consensus. This algorithm,
	however, is not efficient.
	For the case when vertex connectivity is at least $2f$, we have
	a more efficient algorithm, which is presented in the appendices.
	In Section \ref{section efficienct}, we present a tool used in the algorithm which exploits the $2f$-connectivity.
	Note that
	for $f=1,2$, the tight condition presented above implies vertex connectivity of 2 and 4, respectively (i.e., connectivity equal to $2f$ when $f=1,2$).
	
	\item {\bf Hybrid model:}
	In Section \ref{section hybrid}, we extend the above necessary and sufficient condition to a {\em hybrid model} wherein at most $t \le f$ faulty nodes may have the ability to equivocate to their neighbors (i.e., the ability to send messages to each neighbor
	without being overheard by the other neighbors), while the remaining faulty nodes are restricted to local broadcast (i.e., their messages will
	be received identically by all the neighbors).
	The point-to-point communication model and the local broadcast
	model are both obtained as special cases of the hybrid model
	when $t = f$ and $t = 0$, respectively.
	Thus, the hybrid model provides a bridge between those two models and helps to precisely characterize the trade-off between equivocation and network requirements.
	%
	%
%
\end{enumerate}

%
%
%
%

%% file: relatedWork.tex


There is a large body of work on Byzantine fault-tolerant algorithms \cite{Attiya:2004:DCF:983102, DOLEV198214, Lamport:1982:BGP:357172.357176, Lynch:1996:DA:2821576, Pease:1980:RAP:322186.322188}.
Here we focus primarily on related work that imposes constraints on a faulty node's ability to equivocate.

%
%

Rabin and Ben-Or \cite{Rabin:1989:VSS:73007.73014} considered a global broadcast model and showed that $n \ge 2f+1$ is both sufficient and necessary for consensus in a synchronous system.
The local broadcast model considered in our work reduces to the global broadcast model	when the network is a complete graph.
The necessary and sufficient conditions obtained in this paper are (not surprisingly) equivalent to $n\geq 2f+1$ when the network is a complete graph.
Clement et. al. \cite{Clement:2012:PN:2332432.2332490} also considered non-equivocation in a \emph{complete} graph under asynchronous communication.
Our work obtains results for arbitrary {\em incomplete} graphs with synchronous communication.
	
The goal of a Byzantine broadcast algorithm is to allow a single source node to deliver its message reliably to all the other nodes. Prior work has explored such algorithms under the local broadcast model 
\cite{Bhandari:2005:RBR:1073814.1073841, Koo:2004:BRN:1011767.1011807, Koo:2006:RBR:1146381.1146420}.
However, the results for Byzantine {\em broadcast} do not provide insights into the network requirements for Byzantine {\em consensus} problem considered in this paper.

	
There has been significant work on other similar models, sometimes called ``partial broadcast''.
To motivate these models, consider
	a network consisting of several Ethernet channels, with a subset of nodes $S_i$ being connected to channel $i$. Then, transmission by any node in $S_i$ on channel $i$ will be received by all the nodes in channel $i$. Each node may be connected to several different such channels. 
	%
	%
	An Ethernet channel  to which $h$ nodes are connected may be viewed as a $h$-hyperedge in a communication network represented by a hypergraph.
	Fitzi and Maurer \cite{Fitzi:2000:PCG:335305.335363} considered a network in which every subset of three nodes form a hyperedge, in addition to every subset of two nodes also forming a hyperedge (i.e., a complete (2,3)-uniform hypergraph).
	Ravikant et. al. \cite{Ravikant10.1007/978-3-540-30186-8_32} also considered a (2,3)-uniform hypergraph, but with only a subset of 3-hyperedges being present.
	Jaffe et. al. \cite{Jaffe:2012:PEB:2332432.2332491} gave asymptotically tight bounds for the number of $3$-hyperedges required
	in such graphs.
	Amitanand et. al. \cite{Amitanand:2003:DCP:872035.872065} considered a complete network wherein, for each faulty node $k$, the remaining nodes are partitioned such that transmission by the faulty node $k$ to any node is received identically by all the nodes in its partition.
	One key difference from our work is that we consider incomplete networks, whereas \cite{Amitanand:2003:DCP:872035.872065} assumes that each node can communicate directly with all the other nodes.
	Amitanand et. al. \cite{Amitanand:2003:DCP:872035.872065} considered an adversary structure that characterizes the set of nodes that may be simultaneously faulty, instead of using a threshold on the number of faults.
	Other work \cite{Considine2005, Franklin2000, FranklinHypergraphsPrivacy2004, Wang2001} has explored the trade-off between reliability and privacy on partial broadcast networks.

    Some prior work has explored a {\em restricted} class of iterative algorithms that achieve {\em approximate} Byzantine consensus under the local broadcast model \cite{LeBlancIterative, ZhangIterative}.
    In particular, this class of algorithms is iterative in nature,
    with a state variable at each node being updated in each iteration as a linear interpolation of the states of selected neighbors. Due to the restriction on the algorithm behavior, the network requirements exceed the necessary and sufficient conditions
    shown in this paper.
    Additionally, these restricted algorithm structures yield only approximate consensus in finite time.
	Li et. al. \cite{Li7516067} extended this line of work
	to a network consisting of 3-hyperedges and 2-hyperedges.
	Vaidya et. al. \cite{TsengIterativeICDCN, VaidyaICDCN_Vector, Vaidya:2012:IAB:2332432.2332505} have investigated the iterative algorithm structure in the point-to-point communication model.

%% file: model.tex
We consider a synchronous system. The communication network interconnecting $n$ nodes is represented by an undirected graph $G = (V, E)$. Every node in the system knows graph $G$. Each node $u$ is represented by vertex $u\in V$. We will use the terms {\em node} and {\em vertex} interchangeably.
Two nodes $u$ and $v$ are \emph{neighbors} if and only if $uv \in E$ is an edge of $G$.
For a set $S\subseteq V$, node $u$ is said to be a neighbor of set $S$ if $u\not\in S$ and there exists $v\in S$ such that $uv\in E$.


Each edge $uv$ represents a FIFO link between the two nodes $u$ and $v$.
When a message $m$ sent by node $u$ is received by node $v$, node $v$ knows that $m$ was sent by node $u$.
In Sections \ref{section necessity} and \ref{section sufficiency}, we assume the \emph{local broadcast} model wherein a message sent by a node $u$ is received identically and correctly by each node $v$ such that $uv\in E$ (i.e., by each neighbor of $u$).
A {\em hybrid} model is considered later in Section \ref{section hybrid}.

A \emph{Byzantine} faulty node may exhibit arbitrary behavior.
We consider the \emph{Byzantine consensus problem} assuming
that each node has a binary input. The goal is for each node to output a binary value, satisfying the following conditions, in the presence of at most $f < n$ Byzantine faulty nodes.
\begin{enumerate}[label=\arabic*)]
	\item \textbf{Agreement:}
	All non-faulty nodes must output the same value.
	\item \textbf{Validity:}
	The output of each non-faulty node must be the input of some non-faulty node.
	\item \textbf{Termination:}
	All non-faulty nodes must decide on their output in finite time.
\end{enumerate}

\begin{itemize}[wide,labelindent=0pt]
    \item[\textbf{Paths in graph $G$:}]
    A path is a sequence of nodes such that any two adjacent nodes
    in the sequence are neighbors in the graph.
    \begin{itemize}
        \item[\textbullet]
        For $u,v\in V$, a $uv$-path $P_{uv}$ is a path between nodes $u$ and $v$.
        $u$ and $v$ are \emph{endpoints} of path $P_{uv}$.
        Any node in path $P_{uv}$ that is not an endpoint
        is said to be an \emph{internal} node of $P_{uv}$.
        All $uv$-paths have $u$ and $v$ as endpoints, by definition.
        Two $uv$-paths are \emph{node-disjoint} if they do not have any  internal nodes in common.
        
        \item[\textbullet]
        For $U \subset V$ and a node $v \not \in U$, a $Uv$-path is a $uv$-path for some node $u \in U$. All $Uv$-paths have $v$ as one endpoint.
        Two $Uv$-paths are \emph{node-disjoint} if they do not have any nodes in common except endpoint $v$. In particular, although two node-disjoint $Uv$-paths have endpoint $v$ in common, the other endpoints are distinct for the two paths.
    \end{itemize}
    
    \item[\textbf{Fault-free paths:}]
    A path is said to \emph{exclude} a set of nodes $X\subset V$ if no internal node of the path belongs to $X$; however, its endpoints may potentially belong to $X$.
    A path is said to be {\em fault-free} if none of its internal nodes are faulty. In other words, a path is fault-free
    if it excludes the set of faulty nodes. Note that a fault-free path may have a faulty node as an endpoint.
    
    \item[\textbf{Degree and Connectivity:}]
    The \emph{degree} of a node $u$ is the number of $u$'s neighbors (i.e., the number of edges incident to $u$).
    Minimum degree of $G$ is the minimum over the degree of all the vertices in $G$.
    A graph $G$ is $k$-connected if $n > k$ and removal of less than
    $k$ nodes does not disconnect $G$.
    By Menger's Theorem \cite{west2001introduction} a graph $G$ is $k$-connected if and only if for any two nodes $u, v \in V$ there exist $k$ node disjoint $uv$-paths.
    Another standard result \cite{west2001introduction} for $k$-connected graphs is that if $G$ is $k$-connected, then for any node $v$ and a set of at least $k$ nodes $U$ there exist $k$ node-disjoint $Uv$-paths.
\end{itemize}

%% file: necessity.tex

Theorem \ref{theorem necessity} below states the
necessary conditions for Byzantine consensus under the local
broadcast model.
Section \ref{section sufficiency} presents a Byzantine consensus algorithm and constructively proves that the necessary conditions are also sufficient.

\begin{theorem} \label{theorem necessity}
	If there exists a Byzantine consensus algorithm under the local broadcast model on graph $G$ tolerating at most $f$ Byzantine faulty nodes, then both the following conditions must be true:
	\begin{enumerate*}[label=(\roman*)]
	    \item $G$ has minimum degree at least $2f$, and 
	    \item $G$ is $(\floor{3f/2} + 1)$-connected.
	\end{enumerate*}
\end{theorem}

Appendix \ref{section proofs necessity} presents a  proof of Theorem \ref{theorem necessity}.
Necessity of conditions (i) and (ii) in the theorem is proved separately in Lemmas \ref{lemma degree} and \ref{lemma connectivity}, respectively.
We use a state machine based approach \cite{Attiya:2004:DCF:983102, DOLEV198214, Fischer1986}
to prove Lemmas \ref{lemma degree} and \ref{lemma connectivity}.

It should be easy to see that a complete graph consisting of $2f+1$ nodes satisfies the necessary conditions in Theorem \ref{theorem necessity} for any $f$.
Figure \ref{figure graphs} presents other examples of graphs that satisfy the necessary conditions. Figure \ref{figure graphs}(a) shows a cycle consisting of 5 nodes. For this graph, the minimum degree is 2, and the graph is 2-connected. Thus, the cycle satisfies the conditions in Theorem \ref{theorem necessity} for $f=1$.
The graph in Figure \ref{figure graphs}(b) satisfies the necessary conditions when $f=2$.

\begin{figure}[ht]
    \centering
    \begin{subfigure}{0.45\textwidth}
        \centering
        \begin{tikzpicture}
        	\node[draw, circle, minimum height=1cm]
        		at (90:2) (1) {$1$};
        	\node[draw, circle, minimum height=1cm]
        		at (162:2) (2) {$2$};
        	\node[draw, circle, minimum height=1cm]
        		at (234:2) (3) {$3$};
        	\node[draw, circle, minimum height=1cm]
        		at (306:2) (4) {$4$};
        	\node[draw, circle, minimum height=1cm]
        		at (378:2) (5) {$5$};
        	
        	\draw[-] (1) to (2);
        	\draw[-] (2) to (3);
        	\draw[-] (3) to (4);
        	\draw[-] (4) to (5);
        	\draw[-] (5) to (1);
        \end{tikzpicture}
        \caption{$f=1$}
        \label{figure example graph cycle}
    \end{subfigure}
    \hfill
    \begin{subfigure}{0.45\textwidth}
        \centering
        	
        	
        	
        	
        	
        	
        	
        \begin{tikzpicture}
        	\node[draw, circle, minimum height=0.5cm]
        		at (45:1) (1) {};
        	\node[draw, circle, minimum height=0.5cm]
        		at (135:1) (2) {};
        	\node[draw, circle, minimum height=0.5cm]
        		at (225:1) (3) {};
        	\node[draw, circle, minimum height=0.5cm]
        		at (315:1) (4) {};
        	
        	\node[draw, circle, minimum height=0.5cm]
        		at (0:2) (5) {};
        	\node[draw, circle, minimum height=0.5cm]
        		at (90:2) (6) {};
        	\node[draw, circle, minimum height=0.5cm]
        		at (180:2) (7) {};
        	\node[draw, circle, minimum height=0.5cm]
        		at (270:2) (8) {};
        	
        	\draw[-] (5) to (1);
        	\draw[-] (5) to (2);
        	\draw[-] (5) to (3);
        	\draw[-] (5) to (4);
        	
        	\draw[-] (6) to (1);
        	\draw[-] (6) to (2);
        	\draw[-] (6) to (3);
        	\draw[-] (6) to (4);
        	
        	\draw[-] (7) to (1);
        	\draw[-] (7) to (2);
        	\draw[-] (7) to (3);
        	\draw[-] (7) to (4);
        	
        	\draw[-] (8) to (1);
        	\draw[-] (8) to (2);
        	\draw[-] (8) to (3);
        	\draw[-] (8) to (4);
        \end{tikzpicture}
        \caption{$f=2$}
        \label{figure example graph clique}
    \end{subfigure}
    
    \caption{Example graphs satisfying conditions in Theorem \ref{theorem necessity}}
    \label{figure graphs}
\end{figure}
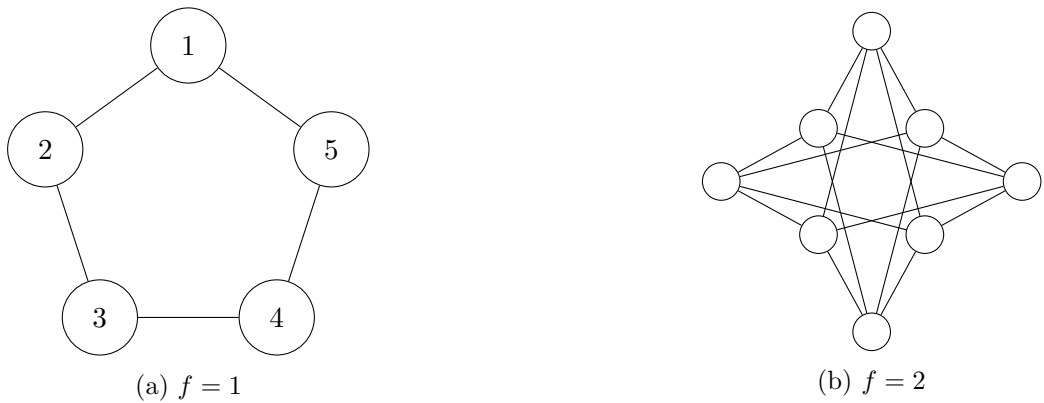

Section \ref{section sufficiency} below proves that the above necessary conditions are also sufficient. Before proceeding to Section \ref{section sufficiency}, let us consider the 5-node cycle in Figure \ref{figure graphs}(a) to build some intuition on why these conditions may be sufficient.
Since $f=1$, we do not have $2f+1$ node-disjoint paths between every pair of nodes in the cycle in Figure \ref{figure graphs}(a).
Despite lower connectivity, we can show a useful property.
In particular, suppose that node 1 attempts to send a message $M$ to node 4, by transmitting it along two node-disjoint paths 1-2-3-4 and 1-5-4, respectively.
That is, node 1 will send the message to nodes 2 and 5 (its neighbors) in a single transmission -- local broadcast model allows node 1 to send its message to all its neighbors simultaneously.
The neighbors 2 and 5 will then forward the message
and subsequently, node 3
will forward the message received from node 2.
Due to the local broadcast model, all neighbors of each node will receive its transmissions. Consider two cases:
\begin{itemize}
    \item
    \underline{Case (i):}
    The internal nodes (namely, 2, 3 and 5) on the two paths behave correctly: In this case, node 4 receives identical copies of the message along the two disjoint paths.
    Because $f=1$, node 4 can be certain that it has correctly received the message that was transmitted by node 1 (when  $f=1$, internal nodes on at most one disjoint path may be faulty).
    
    \item
    \underline{Case (ii):}
    Node 3 is faulty, and tampers the message received from node 2 before forwarding it to node 4: In this case, node 4 will not receive two identical message copies along the two disjoint paths. Therefore, node 4 cannot determine the message sent by node 1. However, in this case, node 2 is non-faulty.
    Node 2 correctly forwards message $M$ received from node 1 to node 3, and then observes that node 3 is forwarding a tampered message to node 4, not the correct message $M$. Node 2 can observe the faulty behavior of node 3 due to the local broadcast property -- when node 3 sends the tampered message to node 4, node 2 receives it too, because node 2 is a neighbor of node 3.
    
    Node 2 can now notify node 1 that node 3 is faulty.
    Of course, on receiving this notification, node 1 cannot be certain whether node 3 is indeed faulty, or node 2 is faulty and it is incorrectly accusing node 3 of misbehavior. However, node 1 can be certain that one of nodes 2 and 3 must be faulty. Because $f=1$, node 1 can then infer that the path 4-5-1 is fault-free. Thus, any messages received by node 1 from node 4 along the path 4-5-1 could not be tampered by the internal node on this path, namely, node 5. This allows node 1 to now receive messages transmitted by node 4 reliably.
\end{itemize}

The example above can be generalized to derive a similar capability for reliable communication in at least one direction between each pair of nodes. Although the algorithm presented next does not seem to explicitly utilize this property, it implicitly relies on such a behavior.

%% file: sufficiency.tex

Theorem \ref{theorem sufficiency} below states that the necessary conditions in Theorem \ref{theorem necessity} are also sufficient.

\begin{theorem} \label{theorem sufficiency}
	Under the local broadcast model, Byzantine consensus tolerating at most $f$ Byzantine faulty nodes is achievable on graph $G$ if both the following conditions are true:
	\begin{enumerate*}[label=(\roman*)]
	    \item $G$ has minimum degree at least $2f$, and 
	    \item $G$ is $(\floor{3f/2} + 1)$-connected.
	\end{enumerate*}
\end{theorem}


We prove correctness of Theorem \ref{theorem sufficiency} constructively by providing a Byzantine consensus algorithm (in Section \ref{section algorithm}) and showing its correctness (in Section \ref{section correctness}).

%% file: algorithm.tex

Assume that graph $G$ satisfies the properties stated
in Theorem \ref{theorem sufficiency}. That is, $G$ has minimum degree at least $2f$ and is $(\floor{ 3f / 2 } + 1)$-connected.
Pseudo-code for the proposed algorithm is presented below.
To understand the description presented next, it will help the reader to read the corresponding steps in Algorithm \ref{algorithm consensus} below.

\begin{itemize}[wide,labelindent=0pt]
    \item[\underline{\em Initialization:}] Each node $v\in V$ maintains a local state variable named $\gamma_v$, which is initialized  to equal node $v$'s binary input.
    
    \item[\underline{\em Phases of the algorithm:}] Recall that we assume a synchronous system. The execution of the algorithm is divided into many {\em phases}, each phase corresponding to a distinct subset of nodes $F$, such that $F$ contains at most $f$ nodes. Each iteration of the {\tt For} loop in Algorithm \ref{algorithm consensus} corresponds to one phase of the algorithm. The set $F$ chosen in each phase is a {\em candidate} for the actual set of faults; however, set $F$ in only one of the phases will exactly equal the actual set of faulty nodes in the given execution.
    This is similar to the Byzantine consensus algorithm for directed graphs under the point-to-point communication model by Tseng and Vaidya \cite{LewisByzantineDirected}.
    However, the rest of the algorithm proceeds differently since we consider the local broadcast model.
    
    \item[\underline{\em Step (a):}] At the beginning of each phase, each node $v$ attempts to communicate its current value of $\gamma_v$ via ``flooding'', as described soon.
    Any message transmitted during flooding has the form $(b, \Pi)$, where $b \in \set{ 0, 1 }$ and $\Pi$ is a path.
    Flooding proceeds in synchronous rounds, with each node possibly forwarding messages received in the previous round, following the rules presented below.
    Flooding will end after $n$ rounds, as should be clear from the following description.
    
    To initiate flooding of its $\gamma_v$ value, node $v$ transmits message $(\gamma_v,\bot)$ to its neighbors, where $\bot$ represents an empty path.
    Each node in the network similarly initiates flooding of its own $\gamma$ value in step (a).
    If $v$ is faulty and does not initiate flooding, then non-faulty neighbors of $v$ replace the missing message with the default message of $(1,\bot)$.
    Therefore, we can assume that a value is indeed flooded by each node, even if it is faulty.
    When node $v$ receives from a neighbor $u$ a message $(b,\Pi)$, where $b\in\{0,1\}$ and $\Pi$ is a path, it takes the following steps sequentially. In the following note that $\Pi \operatorname{-}  u$ denotes a path obtained by concatenating identifier $u$ to path $\Pi$.
    \begin{enumerate}[label=(\roman*)]
        \item If path $\Pi \operatorname{-} u$ does not exist in graph $G$, then node $v$ discards the message $(b,\Pi)$. Recall that each node knows graph $G$, and the message $(b,\Pi)$ was received by node $v$ from node $u$.
        \item Else if, in the current phase, node $v$ has previously received from $u$ another message containing path $\Pi$ (i.e., a message of the form $(b',\Pi)$),
        then node $v$ discards the message $(b,\Pi)$.
        \item Else if path $\Pi$ already includes node $v$'s identifier, node $v$ discards the message $(b,\Pi)$.
        \item Else node $v$ is \underline{said to have received value $b$ along path $\Pi \operatorname{-} u$} and node $v$ forwards message $(b,\Pi \operatorname{-} u)$ to its neighbors. Recall that $v$ received message $(b,\Pi)$ from neighbor $u$.
    \end{enumerate}
    
    Rules (i) and (ii) above are designed to prevent a faulty node from sending spurious messages. Recall that under the local broadcast model, all neighbors of any node receive all its transmissions.
    Thus, rule (ii) above essentially prevents a faulty node from successfully delivering mismatching messages to its neighbors (i.e., this prevents equivocation).
    Rule (iii) ensures that flooding will terminate after $n$ rounds.
    
    Rule (ii) above crucially also ensures the following useful property due to the local broadcast model: even if node $u$ is Byzantine faulty, but paths $P_{uv}$ and $P_{uw}$ are fault-free, then nodes $v$ and $w$ will receive
    identical value in the message from $u$ forwarded along paths $P_{uv}$ and $P_{uw}$, respectively.
    Recall that a path is fault-free if none of the internal nodes on the path are faulty.
    
    \item[\underline{\em Step (b):}] Recall that a path is said to exclude set $F$ if none of its {\em internal} nodes are in $F$. For each $u\in V$, node $v$ chooses an arbitrary $uv$-path $P_{uv}$
    that excludes $F$.
    It can be shown (Lemma \ref{lemma paths}) that such a path always exists under the conditions in Theorem \ref{theorem sufficiency}.
    For the purpose of step (b), node $v$ is deemed to have received its own $\gamma_v$ along
    path $P_{vv}$ (containing only node $v$).
    Node $v$ computes sets $Z_v$ and $N_v = V - Z_v$, as shown in the pseudo-code.
    
    \item[\underline{\em Step (c):}]
    This step specifies the rules for updating value $\gamma_v$.
    $\gamma_v$ is not necessarily updated in each phase.
    
    \item[\underline{\em Output:}]
    After all the phases (i.e., all iterations of the {\tt For} loop) are completed,
    the value of $\gamma_v$ is chosen as the output of node $v$.
\end{itemize}

The proof of correctness of Algorithm \ref{algorithm consensus} is presented in Section \ref{section correctness}. 
%
%
As we discussed earlier when describing the flooding mechanism, the faulty nodes are effectively limited to delivering a unique value on all fault-free paths (i.e., paths that do not have faulty nodes as internal nodes). If this unique value corresponding to a faulty node is $0$, we will say that the faulty node flooded value $0$; else we will say that the node flooded value 1.

Let $Z$ be the set of nodes that flooded $0$ in step (a) and let $N = V - Z$ be the set of the remaining nodes that flooded $1$.
Recall from the description of flooding above, that even a faulty node effectively floods one value, and one value only, in step (a) of each phase.
Note that either $Z$ or $N$ may possibly be empty, but not both.
In step (b), each non-faulty node $v$ obtains its own estimates $Z_v$ and $N_v$ of $Z$ and $V$, respectively. If the set $F$ in the current phase does not contain all the faulty nodes, then these estimates can be incorrect.
However, at least in one phase, $F$ will contain all the faulty nodes and, in that phase, as shown later, it is guaranteed that
$Z=Z_v$ and $N=N_v$.
This observation is important for the correctness of the algorithm, as seen later.



As shown later, step (c) is designed to ensure that the fault-free nodes will reach consensus in a phase in which $F$ contains all the faulty nodes.
In each phase, step (c) also ensures the invariant that $\gamma_v$,
at any non-faulty node $v$, at the end of the phase is equal to $\gamma_u$, for some non-faulty node $u$, at the start of that phase.


\begin{algorithm}[ht]
	\SetAlgoLined
	\SetKwFor{For}{For}{do}{end}
	Each node $v$ has a binary input value in $\set{ 0, 1 }$ and maintains a binary state $\gamma_v \in \set{ 0, 1 }$.\\
	{\em Initialization:} $\gamma_v$ := input value of node $v$\\
	\For{each $F \subseteq V$ such that $\abs{F} \le f$}
	{\begin{enumerate}
			[label=Step (\alph*):,labelindent=12pt,leftmargin=!,rightmargin=\widthof{Step ()}]
			\item
			Flood value $\gamma_v$. (The steps taken to achieve flooding are described in the text preceding Algorithm \ref{algorithm consensus} pseudo-code.)
			
			\item
			For each node $u \in V$, identify a single $uv$-path $P_{uv}$ that excludes $F$.
			Let,
			\begin{align*}
    			Z_v &:= \set{ u \in V \mid \text{$v$ received value $0$ from $u$ along $P_{uv}$ in step (a)} },\\
    			N_v &:= V-Z_v.
			\end{align*}
			
			\item
			Define sets $A_v$ and $B_v$ as follows.
			\begin{enumerate}
				[label=Case \arabic*:]
				\item If $\abs{Z_v \cap F} \le \floor{f/2}$ and $\abs{N_v} > f$, then
				 $A_v := N_v$ and $B_v := Z_v$.
				
				\item If $\abs{Z_v \cap F} \le \floor{f/2}$ and $\abs{N_v} \le f$,
				then $A_v := Z_v$ and $B_v := N_v$.
				
				\item If $\abs{Z_v \cap F} > \floor{f/2}$ and $\abs{Z_v} > f$,
			    then $A_v := Z_v$ and $B_v := N_v$.
				
				\item If $\abs{Z_v \cap F} > \floor{f/2}$ and $\abs{Z_v} \le f$,
				then $A_v := N_v$ and $B_v := Z_v$.
			\end{enumerate}
			If $v \in B_v$ and $v$ receives value $\delta\in\{0,1\}$ along any $f+1$ node-disjoint $A_v v$-paths that exclude $F$ in step (a), then $\gamma_v := \delta$.
	\end{enumerate}}
	
	Output $\gamma_v$.
	
	\caption{
		Proposed algorithm for Byzantine consensus under the local broadcast model: Steps performed by node $v$ are shown here.
	}
	\label{algorithm consensus}
\end{algorithm}

%% file: correctness.tex

In this section, we assume that graph $G$ satisfies the conditions stated in Theorem \ref{theorem sufficiency}, even if this is not always stated explicitly below.
Appendix \ref{section proofs sufficiency} presents the proofs of the lemmas in this section.
%
%

The proof of correctness relies on two key lemmas, Lemma \ref{lemma validity} and \ref{lemma agreement}. We will present additional results and discuss the intuition behind the proof of these lemmas subsequently. Formal proofs are presented in Appendix \ref{section proofs sufficiency}.
Recall that the algorithm execution is divided into {\em phases}, each phase corresponding to a different choice of $F$, where
$|F|\leq f$. For convenience of presentation, we will refer to $\gamma_v$ as the ``state of node $v$''. Recall that state $\gamma_v$ of node $v$ may possibly be modified in step (c) in each phase.
\begin{lemma} \label{lemma validity}
 	For a non-faulty node $v$, its state $\gamma_v$ at the end of any given phase equals the state of some non-faulty node at the start of that phase.
\end{lemma}

\begin{lemma} \label{lemma agreement}
	Consider a phase of Algorithm \ref{algorithm consensus} wherein all the faulty nodes are contained in set $F$ corresponding to that phase. At the end of this phase,
	 every pair of non-faulty nodes $u, v \in V$ have identical state, i.e., $\gamma_u=\gamma_v$.
\end{lemma}

As shown next,
these
two lemmas imply correctness of
Algorithm \ref{algorithm consensus},
thus proving
Theorem \ref{theorem sufficiency}.

\begin{proof_of}{Theorem \ref{theorem sufficiency}}
Algorithm \ref{algorithm consensus} terminates in finite time because the number of phases is finite, and flooding in each phase completes in finite time. Thus, the algorithm satisfies the {\em termination} condition.

Since there are at most $f$ faulty nodes in any
given execution, there exists at least one phase in which
set $F$ will contain all the faulty nodes. Then, from Lemma \ref{lemma agreement}, we have that all non-faulty nodes have the same state at the end of this phase. Lemma \ref{lemma validity} implies that the state of the non-faulty nodes will remain unchanged after any subsequent phases. Therefore, all non-faulty nodes will have the same state at the end of the algorithm and their output will be identical. This proves that the algorithm satisfies the {\em agreement} condition.

At the start of phase 1, the state of each non-faulty node equals its own input. Now, applying Lemma \ref{lemma validity} inductively implies that the state of a non-faulty node always equals the {\em input} of some non-faulty node. This, in turn,
implies that the algorithm satisfies the {\em validity} condition.
Thus, we have proved correctness of Algorithm \ref{algorithm consensus} under the conditions stated in Theorem \ref{theorem sufficiency}.
\end{proof_of}


The proofs of Lemmas \ref{lemma validity} and \ref{lemma agreement}, stated above, rely on two other lemmas, presented next. The reader may skip the rest of this section without a loss of continuity.
%
%
%
Lemma \ref{lemma paths} below implies that path $P_{uv}$ used in step (b) of the algorithm indeed exists.

\begin{lemma} \label{lemma paths}
	For any choice of set $F$ in the algorithm, and any two nodes $u, v\in V$, there exists a $uv$-path that excludes $F$.
\end{lemma}

\begin{proof}
Recall that a path is said to exclude $F$ if none of the
{\em internal} nodes in the path belong to $F$.
%
Since $G$ is $(\floor{3f/2}+1)$-connected, by Menger's Theorem, there are at least $\floor{3f/2}+1$ node-disjoint paths between any two nodes $u, v$.
For any $f\geq 0$, we have
$\floor{3f/2}+1\geq f+1$.
Thus, there are at least $f+1$ node-disjoint $uv$-paths, of which at least one path must exclude $F$, since $\abs{F} \le f$.
%
\end{proof}

The next lemma states that the choice of sets $A_v$ and $B_v$ in step (c) ensures that the node-disjoint paths used in that step indeed exist.

\begin{lemma} \label{lemma propagates}
	For any non-faulty node $v$, and any given phase with the corresponding set $F$, in step (c), if $v \in B_v$, then there exist $f+1$ node-disjoint $A_v v$-paths that exclude $F$.
\end{lemma}

The formal proof of the above lemma is presented in Appendix \ref{section proofs sufficiency}. Observe that there are four distinct cases in step (c) for determining sets $A_v$ and $B_v$. The proof of the above claim in cases 1 and 3 in step (c) follows from $(\floor{3f/2}+1)$-connectivity of graph $G$, while the proof in cases 2 and 4 follows from the fact that the minimum degree
of $G$ is at least $2f$.

In step (c), observe that if node $v$ modifies its state $\gamma_v$, then
it must have received identical value along $f+1$
node-disjoint $A_v v$ paths.
So, at least one of these path must not only be fault-free, but also have
non-faulty endpoints.
The proof of Lemma \ref{lemma validity} in Appendix \ref{section proofs sufficiency} relies on this observation.

Now consider Lemma \ref{lemma agreement}.
The correctness of this lemma relies on the local broadcast property and rule (ii) used in flooding.
Suppose that set $F$ in a certain phase contains all the faulty nodes. Since the paths used in step (b) of the algorithm exclude set $F$, these paths
are fault-free (i.e., none of their internal nodes are faulty).
Then, the earlier discussion of flooding implies that any two non-faulty nodes $u,v$ will obtain $Z_u=Z_v$ and $N_u=N_v$ in step (b) of this phase. By a similar argument, all the paths used in step (c) of this phase are also fault-free, and any two non-faulty nodes will end step (c) of this phase with an identical state. A complete proof of Lemma \ref{lemma agreement} is presented in Appendix \ref{section proofs sufficiency}.

%% file: efficient.tex
The number of phases in Algorithm \ref{algorithm consensus} is exponential in $f$, since there exists one phase corresponding to each choice of set $F$ such that $|F|\leq f$.
When the communication graph $G$ is $2f$-connected, we have developed an efficient algorithm that requires $O(n)$ time.
Although in general $2f$-connectivity is a stronger requirement on graph $G$ as compared to the requirement in Theorem \ref{theorem sufficiency}, observe that when $f=1,2$ these two requirements are equivalent.

\begin{theorem} \label{theorem efficient algorithm}
	Under the local broadcast model, if $G$ is $2f$-connected, then Byzantine consensus tolerating at most $f$ Byzantine faults is achievable on graph $G$ in $O(n)$ synchronous rounds.
\end{theorem}

We prove Theorem \ref{theorem efficient algorithm} constructively by giving an efficient algorithm, which we present in Appendix \ref{section efficient algorithm}.
Using the example in Figure \ref{figure graphs}(a), we now illustrate a tool used in that algorithm, which exploits the $2f$-connectivity.
The graph in Figure \ref{figure graphs}(a) is $2$-connected, i.e., $2f$-connected for $f=1$.
Similar to our algorithm in the previous section, suppose that node 1 floods value $b$, which is propagated to node 4 along path 1-2-3-4.
When node 2 forwards the value to node 3, all of node 2's neighbors hear the forwarded message. Similarly, when node 3 forwards the message to node 4, all of node 3's neighbors hear the forwarded message.

Suppose now that we ask each node to additionally ``report'' on its neighbors by flooding any messages heard/received in the above propagation from node 1 to node 5. Then the message forwarded by node 4, as overheard by node 3, will be flooded by node 3. Since the graph $G$ is $2f$-connected, node 4 has at least $2f$ neighbors, each of which will take similar steps.  Now consider two cases:
\begin{itemize}
    \item Node 3 is faulty:
    Since node 3 is faulty, there are at most $f-1$ other faulty nodes.
    In this case, $2f$-connectivity implies that there are at least $f+1$ node-disjoint fault-free paths from neighbors of node 3 to node 1.
    Thus, node 1 can receive $f+1$ identical and correct reports of all messages transmitted by node 3.
    Therefore, node 1 can correctly determine the message forwarded by node 3.
    In general, each node can correctly learn messages transmitted by any faulty node, when the connectivity is $2f$.
    
    \item Node 3 is non-faulty:
    In this case, it is possible that node 1 does not receive identical reports about node 3's messages on $f+1$ node-disjoint paths.
    This inability to receive $f+1$ identical reports, however, allows node 1 to infer that node 3 must be non-faulty.
\end{itemize}
In summary, as illustrated above, each node can observe all messages sent by any faulty node. Also, each node can either observe all messages sent by
another
non-faulty node, or learn that it is non-faulty.
In some instances, these observations help non-faulty nodes
identify
the
faulty nodes accurately.
Our algorithm in Appendix \ref{section efficient algorithm} relies on this property to improve the time complexity.

%% file: hybridFaults.tex

In this section we consider a hybrid model. The hybrid model is designed to help explore the gap between the network requirements for Byzantine consensus under the point-to-point communication model and the local broadcast model.
Under the hybrid model, up to $f < n$ nodes may be Byzantine faulty. The faulty nodes are of two types:
\begin{itemize}
  \item Equivocating faulty nodes: At most $t \le f$ of the faulty nodes may equivocate, that is, they are not restricted to perform local broadcast. If $u$ is an equivocating faulty node and has neighbors $v$ and $w$,
  then node $u$ may send message $M_v$ to node $v$ without node $w$ receiving $M_v$, and similarly send message $M_w$ to node $w$ without node $v$ receiving $M_w$. Thus, equivocating faulty nodes can behave similar to the faulty nodes under the point-to-point communication model.
  \item Non-equivocating faulty nodes: Any faulty node that is not an equivocating faulty node is said to be a non-equivocating faulty node. A non-equivocating faulty node conforms to the local broadcast model. Thus, if node $u$ is a non-equivocating faulty node, and has neighbors $v$ and $w$,
  then any message $M$ transmitted by node $u$ will be received identically by $v$ and $w$ both.
\end{itemize}
Observe that when $t = 0$, the hybrid model reduces to the local broadcast model, since all the faulty nodes are restricted to perform local broadcast. On the other hand, when $t = f$, the hybrid model reduces to the classical point-to-point communication model because all the faulty nodes can equivocate.
The following theorem extends results from
Sections \ref{section necessity} and \ref{section sufficiency}
to the hybrid model.

\begin{theorem} \label{theorem hybrid}
	Under the hybrid model, Byzantine consensus tolerating at most $f$ Byzantine faulty nodes, of which at most $t$ are equivocating faulty nodes, is achievable on graph $G$ if and only if all the following conditions are true:
	\begin{enumerate}
		[label=(\roman*)]
		\item $G$ is $(\floor{3(f-t)/2} + 2t + 1)$-connected,
		\item if $t = 0$, then $G$ has minimum degree at least $2f$, and
		\item if $t > 0$, then every set of nodes $S$, such that $0 < |S| \leq t$, has at least $2f + 1$ neighbors.
	\end{enumerate}
\end{theorem}

Observe that when $t=0$, condition (ii) lower bounds the number of neighbor of each vertex.
On the other hand, when $t > 0$, condition (iii) lower bounds the number of neighbors of each subset of nodes of size at most $t$.
Recall that neighbors of $S$ are nodes outside of $S$ that have an edge to some node in $S$.
Theorem \ref{theorem hybrid} is proved in Appendix \ref{section proofs hybrid}. Consider three cases:
\begin{itemize}
    \item When $t = 0$, as noted earlier, the hybrid model reduces to the local broadcast model. In this case, condition (iii) imposes no restrictions, and conditions (i) and (ii) reduce to the graph requirements in Theorems \ref{theorem necessity} and  \ref{theorem sufficiency}.
    \item When $t = f$, the hybrid model reduces to the point-to-point communication model. In this case, condition (ii) imposes no restrictions, condition (i) requires $G$ to be $(2f+1)$-connected, and condition (iii) implies $n \ge 3f+1$.
    \item When $0 < t < f$, the above theorem provides insights into the trade-off between equivocation and network requirements.
\end{itemize}
    
The necessity and sufficiency of the conditions in Theorem \ref{theorem hybrid} is proved using similar techniques as the proofs of Theorems \ref{theorem necessity} and \ref{theorem sufficiency}, respectively.
The proof of Theorem \ref{theorem hybrid} appears in Appendix \ref{section proofs hybrid}.
To prove sufficiency, Algorithm \ref{algorithm consensus} is modified to obtain an algorithm for the hybrid model -- the modified algorithm is presented in Appendix \ref{section proofs hybrid sufficiency}.

%% file: discussion.tex
In this work, we investigated Byzantine Consensus under the local broadcast model.
We showed that $(\floor{3f/2} + 1)$-connectivity and  minimum degree at least $2f$ are together necessary and sufficient conditions to achieve Byzantine consensus in the presence of at most $f$ Byzantine faults under the local broadcast model.
%
%
The sufficiency proof is constructive.
However, the algorithm presented requires exponential synchronous rounds.
For a stronger network condition of $2f$-connectivity, an efficient algorithm that achieves consensus in linear number of rounds is presented in Appendix \ref{section efficient algorithm}.
We leave finding an efficient algorithm for the tight condition for future work.

We also considered a hybrid model where some faulty nodes may equivocate but other faulty nodes are restricted to local broadcast.
We presented necessary and sufficient conditions for this model, which provide insights into the trade-off between equivocation and network requirements.


%% file: proofsNecessity.tex
In the appendices, with a slight abuse of terminology, we allow a \emph{partition} of a set to have empty parts.
That is, $(Z_1, \dots, Z_k)$ is a partition of a set $Y$ if $\bigcup_{i=1}^{k} Z_i = Y$ and $Z_i \cap Z_j = \emptyset$ for all $i \ne j$, but some $Z_i$'s can be possibly empty.

A Byzantine Consensus algorithm $\mathcal{A}$ outlines a procedure $\mathcal{A}_u$ for each node $u \in V$ that describes state transitions of $u$.
In each synchronous round, each node optionally sends messages to its neighbors, receives messages from the neighbors, and then updates its state.
The new state of $u$ depends entirely on $u$'s previous state and the messages received by $u$ from its neighbors.
The state of $u$ determines the messages sent by $u$.

\input{degree.tex}

\input{connectivity.tex}

~

\begin{proof_of}{Theorem \ref{theorem necessity}}
	Directly from Lemmas \ref{lemma degree} and \ref{lemma connectivity}.
\end{proof_of}

%% file: degree.tex
%
\begin{lemma} \label{lemma degree}
	If there exists a Byzantine consensus algorithm under the local broadcast model on a graph $G$ tolerating at most $f$ Byzantine faulty nodes, then $G$ has minimum degree at least $2f$.
\end{lemma}
%
\begin{proof}
	When $f = 0$, the lemma does not impose any restrictions on $G$.
	So we assume that $f > 0$.
	It is easy to show that, when $n > 1$, each node must have at least one neighbor to be able to achieve consensus.
	So in the rest of the proof we assume that the degree of each node in $G$ is at least $1$.
	Suppose for the sake of contradiction that there exists a node $z$ in $G$ of degree less than $2f$ and there exists an algorithm $\mathcal{A}$ that solves Byzantine consensus under the local broadcast model on $G$.
	Then there exists a partition $(F^1, F^2)$ of the neighborhood of $z$ such that $\abs{F^1} < f$ and $\abs{F^2} \le f$.
	Let $W = V - (F^1 \cup F^2 \cup \set{z})$ be the set of remaining nodes.
	Note that some of these sets can be possibly empty.
	However, since $n > f > 0$ and $z$ has degree at least $1$, we select these sets so that $F^2$ is necessarily non-empty.
	Recall that $\mathcal{A}$ outlines a procedure $\mathcal{A}_u$ for each node $u$ that describes $u$'s state transitions in each round.
	
	We first create a network $\mathcal{G}$ to model behavior of nodes in $G$ in three different executions $E_1$, $E_2$, and $E_3$, which we will describe later.
	Figure \ref{figure degree network} depicts $\mathcal{G}$.
	The network $\mathcal{G}$ has some directed edges, the behavior of which will be explained later.
	We denote a directed edge from $u$ to $v$ as $\overrightarrow{uv}$.
	$\mathcal{G}$ consists of two copies of each node in $W$ and a single copy of the remaining nodes.
	We denote the two sets of copies of $W$ as $W_0$ and $W_1$.
	For each node $u \in W$, we denote by $u_0$ and $u_1$ the two copies of $u$ in $W_0$ and $W_1$ respectively.
	For each edge $uv \in E(G)$, we create edges in $\mathcal{G}$ as follows:
	\begin{enumerate}[label=\arabic*),topsep=0pt,labelindent=0pt,itemsep=0pt]
		\item If $u, v \in W$, then there are two edges $u_0 v_0$ and $u_1 v_1$ in $\mathcal{G}$.
		These edges are not shown in Figure \ref{figure degree network}.
		\item If $u, v \in V - W$, then there is a single edge $uv$ in $\mathcal{G}$.
		Some of these edges are also not shown in Figure \ref{figure degree network}.
		\item If $u \in F^1$ and $v \in W$, then there are two edges $uv_0$ and $\overrightarrow{u v_1}$ in $\mathcal{G}$.
		In Figure \ref{figure degree network}, these edges are illustrated by a single undirected edge between sets $F^1$ and $W_0$, and a single directed edge from $F^1$ to $W_1$.
		\item If $u \in F^2$ and $v \in W$, then there are two edges $\overrightarrow{uv_0}$ and $u v_1$ in $\mathcal{G}$.
		In Figure \ref{figure degree network}, these edges are illustrated by a single directed edge from $F^2$ to $W_1$, and a single undirected edge between sets $F^2$ and $W_0$.
	\end{enumerate}
	Note that there are no edges between $W_0$ and $W_1$.
	In Figure \ref{figure degree network}, this is emphasized by drawing a cross on a dotted line between $W_0$ and $W_1$.
	
	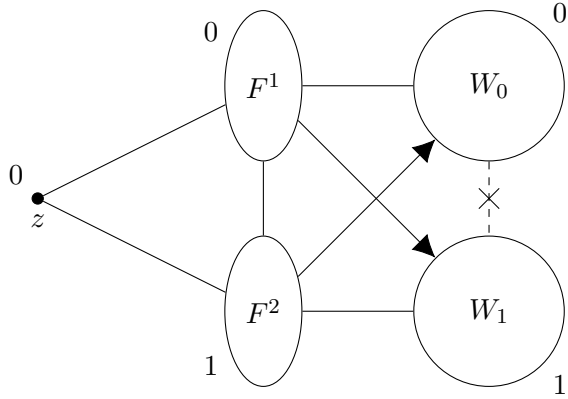
\begin{figure}[t]
		\centering
		\begin{tikzpicture}
		\node[draw, circle, fill, inner sep=1.5pt, label={below:$z$}, 
		label={above left:$0$}] 
		at (0, 0) (z) {};
		\node[draw, ellipse, minimum height=2cm, label={above left:$0$}] 
		at (3, 1.5) (F^1) {$F^1$};
		\node[draw, ellipse, minimum height=2cm, label={below left:$1$}] at (3, -1.5) (F^2) {$F^2$};
		\node[draw, circle, minimum size=2cm, label={above right:$0$}] at (6, 1.5) (W_0) {$W_0$};
		\node[draw, circle, minimum size=2cm, label={below right:$1$}] at (6, -1.5) (W_1) {$W_1$};
		
		\draw[dashed] (W_0) to node[cross out,draw,solid]{} (W_1);
		
		\draw[-{Latex[width=3mm,length=3mm]}] (F^1) to (W_1);
		\draw[-{Latex[width=3mm,length=3mm]}] (F^2) to (W_0);
		
		\draw[-] (z) to (F^1);
		\draw[-] (z) to (F^2);
		\draw[-] (F^1) to (F^2);
		\draw[-] (F^1) to (W_0);
		\draw[-] (F^2) to (W_1);
		\end{tikzpicture}
		\caption{
			Network $\mathcal{G}$ to model executions $E_1$, $E_2$, and $E_3$.
			The edges in $\mathcal{G}$ are described in the text.
			Edges within the sets are not shown while edges between sets/nodes are depicted as single edges.
			The crossed dotted line between $W_0$ and $W_1$ emphasizes that there are no edges between $W_0$ and $W_1$.
			The numbers adjacent to the sets/nodes are the corresponding inputs in execution $\mathcal{E}$.
			Table \ref{table degree} illustrates which nodes in $\mathcal{E}$ model the corresponding nodes in $E_1$, $E_2$, and $E_3$.
		}
		\label{figure degree network}
	\end{figure}
	
	All message transmissions in $\mathcal{G}$ are via local broadcast, as follows.
	When a node $u$ in $\mathcal{G}$ transmits a message, the following nodes receive this message identically: each node with whom $u$ has an undirected edge and each node to whom there is an edge directed away from $u$.
	Note that a directed edge $e = \overrightarrow{uv}$ behaves differently for $u$ and $v$.
	All messages sent by $u$ are received by $v$.
	No message sent by $v$ is received by $u$.
	Observe that with this behavior of directed edges, the structure of $\mathcal{G}$ ensures the following property.
	For each edge $uv$ in the original graph $G$, each copy of $u$ receives messages from exactly one copy of $v$ in $\mathcal{G}$.
	This allows us to create an algorithm for $\mathcal{G}$ corresponding to $\mathcal{A}$ as follows.
	For each node $u \in G$, if $\mathcal{G}$ has one copy of $u$, then $u$ runs $\mathcal{A}_u$.
	Otherwise there are two copies $u_0$ and $u_1$ of $u$.
	Both $u_0$ and $u_1$ run $\mathcal{A}_u$.
	
	Consider an execution $\mathcal{E}$ of the above algorithm on $\mathcal{G}$ as follows.
	Each node in $W_0 \cup F^1 \cup \set{z}$ has input $0$ and the remaining nodes have input $1$.
	Observe that it is not guaranteed that nodes in $\mathcal{G}$ will decide on the same value or that the algorithm will terminate.
	We will show that the algorithm does indeed terminate but nodes do not reach agreement in $\mathcal{G}$, which will be useful in deriving the desired contradiction.
	We use $\mathcal{E}$ to describe three executions $E_1$, $E_2$, and $E_3$ of $\mathcal{A}$ on the original graph $G$ as follows (see also Table \ref{table degree}).
	\begin{enumerate}[label=$E_{\arabic*}$:,topsep=0pt,labelindent=0pt,itemsep=0pt]
		\item
		$F^2$ is the set of faulty nodes.
		In each round, a faulty node broadcasts the same messages as the corresponding node in $\mathcal{G}$ in execution $\mathcal{E}$ in the same round.
		All non-faulty nodes have input $0$.
		Note that the behavior of non-faulty nodes in $F^1 \cup \set{z}$ and $W$ is modelled by the corresponding (copies of) nodes in $F^1 \cup \set{z}$ and $W_0$, respectively, in $\mathcal{E}$.
        Since $\mathcal{A}$ solves Byzantine consensus on $G$, nodes in $W \cup F^1 \cup \set{z}$ decide on output $0$ (by validity) in finite time.
		
		\item
		$F^1$ is the set of faulty nodes.
		In each round, a faulty node broadcasts the same messages as the corresponding node in $\mathcal{G}$ in execution $\mathcal{E}$ in the same round.
		$z$ has input $0$ and all the remaining non-faulty nodes have input $1$.
		Note that the behavior of non-faulty nodes in $F^2 \cup \set{z}$ and $W$ is modelled by the corresponding (copies of) nodes in $F^2 \cup \set{z}$ and $W_1$, respectively, in $\mathcal{E}$.
        The output of the non-faulty nodes will be described later.
		
		\item
		$F^1 \cup \set{z}$ is the set of faulty nodes.
		In each round, a faulty node broadcasts the same messages as the corresponding node in $\mathcal{G}$ in execution $\mathcal{E}$ in the same round.
		All non-faulty nodes have input $1$.
		Note that the behavior of non-faulty nodes in $F^2$ and $W$ is modelled by the corresponding (copies of) nodes in $F^2$ and $W_1$, respectively, in $\mathcal{E}$.
        Since $\mathcal{A}$ solves Byzantine consensus on $G$, nodes in $W \cup F^2$ decide on output $1$ (by validity) in finite time.
	\end{enumerate}
	
	{
		\setlength{\tabcolsep}{18pt}
		\renewcommand{\arraystretch}{1.5}
		\begin{table}[t]
			\centering
			\begin{tabular}{|c|c|c|c|c|c|}
				\hline
				\multicolumn{2}{|c|}{}
				& \multicolumn{4}{|c|}{(Sets of) nodes in network $\mathcal{G}$}\\
				\cline{3-6}
				\multicolumn{2}{|c|}{} & $z$ & $F^1$ & $F^2$ & $W$ \\
				\hline
				\multirow{3}{*}{Executions on graph $G$}& $E_1$ & $z$ & $F^1$ & {\color{red}$F^2$} & $W_0$ \\
				\cline{2-6}
				& $E_2$ & $z$ & {\color{red}$F^1$} & $F^2$ & $W_1$ \\
				\cline{2-6}
				& $E_3$ & {\color{red}$z$} & {\color{red}$F^1$} & $F^2$ & $W_1$ \\
				\hline
			\end{tabular}
			\caption{
				This table illustrates which nodes in execution $\mathcal{E}$ on network $\mathcal{G}$ (Figure \ref{figure degree network}) model the corresponding nodes in the three executions $E_1$, $E_2$, and $E_3$ on graph $G$ for proof of Lemma \ref{lemma degree}.
				The entries in red indicate faulty nodes in $E_1$, $E_2$, and $E_3$.
			}
			\label{table degree}
		\end{table}
	}
	
	Due to the output of nodes in $W \cup F^1 \cup \set{z}$ in $E_1$, the nodes in $W_0 \cup F^1 \cup \set{z}$ output $0$ in $\mathcal{E}$.
	Similarly, due to the output of nodes in $W \cup F^2$ in $E_3$, the nodes in $W_1 \cup F^2$ output $1$ in $\mathcal{E}$.
	It follows that in $E_2$, nodes in $W$ and $F^2$, as modeled by $W_1$ and $F^2$ in $\mathcal{E}$, output $1$ while $z$ outputs $0$.
	Recall that, by construction, $F^2$ is non-empty.
	This violates agreement, a contradiction.
\end{proof}

%% file: connectivity.tex
\begin{lemma} \label{lemma connectivity}
	If there exists a Byzantine consensus algorithm under the local broadcast model on a graph $G$ tolerating at most $f$ Byzantine faulty nodes, then $G$ is $(\floor{3f/2} + 1)$-connected.
\end{lemma}

\begin{proof}
	Suppose for the sake of contradiction that $G$ is not $(\floor{3f/2} + 1)$-connected and there exists an algorithm $\mathcal{A}$ that solves Byzantine consensus under the local broadcast model on $G$.
	Then there exists a vertex cut $C$ of $G$ of size at most $\floor{3f/2}$ with a partition $(A, B, C)$ of $V$ such that $A$ and $B$ (both non-empty) are disconnected in $G - C$ (so there is no edge between a node in $A$ and a node in $B$).
	Since $\abs{C} \le \floor{3f/2}$, there exists a partition $(C^1, C^2, C^3)$ of $C$ such that $\abs{C^1}, \abs{C^2} \le \floor{f/2}$ and $\abs{C^3} \le \ceil{f/2}$.
	Recall that $\mathcal{A}$ outlines a procedure $\mathcal{A}_u$ for each node $u$ that describes $u$'s state transitions in each round.
	
	Similar to the proof of Lemma \ref{lemma degree}, we first create a network $\mathcal{G}$ to model behavior of nodes in $G$ in three different executions $E_1$, $E_2$, and $E_3$, which we will describe later.
	$\mathcal{G}$ consists of two copies of each node in $A$ and $B$, and a single copy of the remaining nodes.
	Figure \ref{figure connectivity network} depicts $\mathcal{G}$.
	The edges in $\mathcal{G}$ can be deduced from Figure \ref{figure connectivity network} along the lines of the proof of Lemma \ref{lemma degree}.
	However, edges within $C$ are not shown in Figure \ref{figure connectivity network}.
	They are copied exactly from $G$.
	As in proof of Lemma \ref{lemma degree}, the structure of $\mathcal{G}$ ensures that, for each edge $uv$ in the original graph $G$, each copy of $u$ receives messages from exactly one copy of $v$ in $\mathcal{G}$.
	This allows us to create an algorithm for $\mathcal{G}$ corresponding to $\mathcal{A}$ as follows.
	For each node $u \in G$, if $\mathcal{G}$ has one copy of $u$ then $u$ runs $\mathcal{A}_u$.
	Otherwise there are two copies $u_0$ and $u_1$ of $u$.
	Both $u_0$ and $u_1$ run $\mathcal{A}_u$.
	
	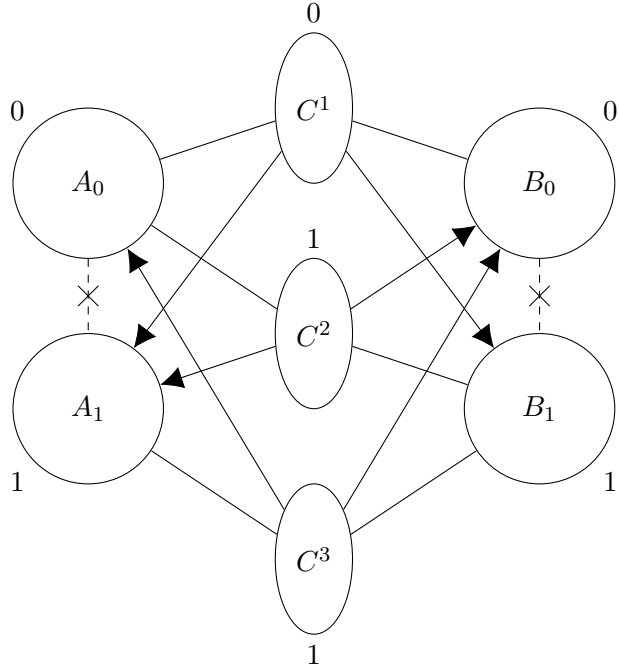
\begin{figure}
		\centering
		\begin{tikzpicture}
			\node[draw, circle, minimum size=2cm, label={above left:$0$}] at (0, 0) (A_0) {$A_0$};
			\node[draw, circle, minimum size=2cm, label={above right:$0$}] at (6, 0) (B_0) {$B_0$};
			\node[draw, circle, minimum size=2cm, label={below left:$1$}] at (0, -3) (A_1) {$A_1$};
			\node[draw, circle, minimum size=2cm, label={below right:$1$}] at (6, -3) (B_1) {$B_1$};
			\node[draw, ellipse, minimum height=2cm, label={above:$0$}] at (3, 1) (C^1) {$C^1$};
			\node[draw, ellipse, minimum height=2cm, label={above:$1$}] at (3, -2) (C^2) {$C^2$};
			\node[draw, ellipse, minimum height=2cm, label={below:$1$}] at (3, -5) (C^3) {$C^3$};
			
			\draw[dashed] (A_0) to node[cross out,draw,solid]{} (A_1);
			\draw[dashed] (B_0) to node[cross out,draw,solid]{} (B_1);
			
			\draw[-{Latex[width=3mm,length=3mm]}] (C^1) to (B_1);
			\draw[-{Latex[width=3mm,length=3mm]}] (C^1) to (A_1);
			\draw[-{Latex[width=3mm,length=3mm]}] (C^2) to (B_0);
			\draw[-{Latex[width=3mm,length=3mm]}] (C^2) to (A_1);
			\draw[-{Latex[width=3mm,length=3mm]}] (C^3) to (B_0);
			\draw[-{Latex[width=3mm,length=3mm]}] (C^3) to (A_0);
			
			\draw[-] (C^1) to (A_0);
			\draw[-] (C^1) to (B_0);
			\draw[-] (C^2) to (A_0);
			\draw[-] (C^2) to (B_1);
			\draw[-] (C^3) to (A_1);
			\draw[-] (C^3) to (B_1);
%
		\end{tikzpicture}
		\caption{
			Network $\mathcal{G}$ to model executions $E_1$, $E_2$, and $E_3$.
			Edges within the sets are not shown while edges between sets/nodes are depicted as single edges.
			The crossed dotted lines emphasize that there are no edges between the corresponding sets.
			The numbers adjacent to the sets/nodes are the corresponding inputs in execution $\mathcal{E}$.
			To reduce clutter, edges within $C$ are not shown.
		}
		\label{figure connectivity network}
	\end{figure}
	
	Consider an execution $\mathcal{E}$ of the above algorithm on $\mathcal{G}$ as follows.
	Each node in $A_0$, $B_0$, and $C^1$ has input $0$ and the remaining nodes have input $1$.
	As in proof of Lemma \ref{lemma degree}, we will show that the algorithm does indeed terminate but nodes do not reach agreement in $\mathcal{G}$, which will be useful in deriving the desired contradiction.
	We use $\mathcal{E}$ to describe three executions $E_1$, $E_2$, and $E_3$ of $\mathcal{A}$ on the original graph $G$ as follows.
	\begin{enumerate}[label=$E_{\arabic*}$:,topsep=0pt,labelindent=0pt,itemsep=0pt]
		\item
		$C^2 \cup C^3$ is the set of faulty nodes.
		In each round, a faulty node in $C^2 \cup C^3$ broadcasts the same messages as the corresponding node in $\mathcal{G}$ in execution $\mathcal{E}$ in the same round.
		All non-faulty nodes have input $0$.
		Note that the behavior of non-faulty nodes in $A$, $B$, and $C^1$ is modelled by the corresponding (copies of) nodes in $A_0$, $B_0$, and $C^1$, respectively, in $\mathcal{E}$.
        Since $\mathcal{A}$ solves Byzantine consensus on $G$, nodes in $A$, $B$, and $C^1$ decide on output $0$ (by validity) in finite time.
		
		\item
		$C^1 \cup C^3$ is the set of faulty nodes.
		In each round, a faulty node in $C^1 \cup C^3$ broadcasts the same messages as the corresponding node in $\mathcal{G}$ in execution $\mathcal{E}$ in the same round.
		$A$ has input $0$ and all the remaining non-faulty nodes have input $1$.
		Note that the behavior of non-faulty nodes in $A$, $B$, and $C^2$ is modelled by the corresponding (copies of) nodes in $A_0$, $B_1$, and $C^2$, respectively, in $\mathcal{E}$.
        The output of the non-faulty nodes will be described later.
		
		\item
		$C^1 \cup C^2$ is the set of faulty nodes.
		In each round, a faulty node in $C^1 \cup C^2$ broadcasts the same messages as the corresponding node in $\mathcal{G}$ in execution $\mathcal{E}$ in the same round.
		All non-faulty nodes have input $1$.
		Note that the behavior of non-faulty nodes in $A$, $B$, and $C^3$ is modelled by the corresponding (copies of) nodes in $A_1$, $B_1$, and $C^3$, respectively, in $\mathcal{E}$.
        Since $\mathcal{A}$ solves Byzantine consensus on $G$, nodes in $A$, $B$, and $C^3$ decide on output $1$ (by validity) in finite time.
	\end{enumerate}
	
    Due to the output of nodes in $A$, $B$, and $C^1$ in $E_1$, the nodes in $A_0$, $B_0$, and $C^1$ output $0$ in $\mathcal{E}$.
    Similarly, due to the output of nodes in $A$, $B$, and $C^3$ in $E_3$, the nodes in $A_1$, $B_1$, and $C^3$ output $1$ in $\mathcal{E}$.
	It follows that in $E_2$, nodes in $A$, as modeled by $A_0$ in $\mathcal{E}$, output $0$ while nodes in $B$, as modeled by $B_1$, output $1$.
	Recall that, by construction, both $A$ and $B$ are non-empty.
	This violates agreement, a contradiction.
\end{proof}

%% file: proofsSufficiency.tex
In this section, we assume that the graph $G$ satisfies the conditions in Theorem \ref{section sufficiency}.
The following observation follows from the rules used for flooding.

\begin{observation} \label{observation fault-free reliable}
	For any phase of Algorithm \ref{algorithm consensus}, in step (a) for any two nodes $u, v \in V$ (possibly faulty), if $v$ receives value $b$ along a fault-free $uv$-path then $u$ broadcast the value $b$ to its neighbors during flooding.
\end{observation}

Recall that, if a faulty node does not initiate flooding, then for the purpose of the above observation, its behavior is equivalent to it flooding the value $1$.
We now give proofs of Lemmas \ref{lemma propagates}, \ref{lemma validity}, and \ref{lemma agreement}.

\begin{proof_of}{Lemma \ref{lemma propagates}}
	Fix a phase of the algorithm and the corresponding set $F$.
	Consider an arbitrary non-faulty node $v$ such that $v \in B_v$ in step (c).
	There are 4 cases to consider, corresponding to the 4 cases in step (c).
	\begin{enumerate}
		[label=Case \arabic*:,topsep=0pt,labelindent=0pt]
		\item $\abs{Z_v \cap F} \le \floor{f/2}$ and $\abs{N_v} > f$.
		Then $A_v := N_v$ and $B_v := Z_v$.
		Therefore there exist at least $f+1$ nodes in $A_v$.
		Node $v$ selects $f+1$ nodes $A'_v$ from $A_v$ by choosing all nodes from $A_v \cap F$ and the rest arbitrarily from $A_v - F$.
        Define $B'_v := B_v \cap (F - v)$.
        Now, $\abs{B'_v} \le \abs{B_v \cap F} = \abs{Z_v \cap F} \le \floor{f/2}$ by assumption of case 1.
		Since $G$ is $(\floor{3f/2}+1)$-connected, $G - B'_v$ is at least $(f+1)$-connected and there exist $f+1$ node-disjoint $A'_v v$-paths in $G - B'_v$.
		Furthermore, since all the nodes in $A_v \cap F$ are endpoints in these paths and $F = (A_v \cap F) \cup (B_v \cap F)$, we have that these paths exclude $F$\footnote{Recall that a path that excludes $X$ does not have nodes from $X$ as internal nodes; however, nodes from $X$ may be the endpoints of the path.}.
		
		\item $\abs{Z_v \cap F} \le \floor{f/2}$ and $\abs{N_v} \le f$.
		Then $A_v := Z_v$ and $B_v := N_v$.
		Note that when $f = 0$, this case is not possible since $v \in B_v$ and so $B_v = N_v$ must be non-empty.
		Since the degree of $v$ is at least $2f$ and there are at most $f$ nodes in $B_v$ (including $v$), we have that $v$ has at least $f+1$ neighbors in $A_v$.
		There are therefore $f+1$ node-disjoint $A_v v$-paths that have no internal nodes and hence exclude $F$.
		
		\item $\abs{Z_v \cap F} > \floor{f/2}$ and $\abs{Z_v} > f$.
		Then $A_v := Z_v$ and $B_v := N_v$.
		We have that
		\begin{align*}
		\abs{N_v \cap F}
		&=      \abs{F} - \abs{Z_v \cap F}  \\
		&\le    f - \floor{f/2} - 1 \\
		&\le    \floor{f/2}.
		\end{align*}
		So this case is the same as Case 1 with the roles of $Z_v$ and $N_v$ swapped.
		
		\item $\abs{Z_v \cap F} > \floor{f/2}$ and $\abs{Z_v} \le f$.
		Then $A_v := N_v$ and $B_v := Z_v$.
		From the analysis in Case 3, we have that $\abs{N_v \cap F} \le \floor{f/2}$.
		So this case is the same as Case 2 with the roles of $Z_v$ and $N_v$ swapped.
	\end{enumerate}
	In all four cases we have that there do exist $f+1$ node-disjoint $A_v v$-paths that exclude $F$.
\end{proof_of}\\

\begin{proof_of}{Lemma \ref{lemma validity}}
	Fix a phase of the algorithm and the corresponding set $F$.
	For any node $u$, we denote the state at the beginning of the phase by $\gamma_u^{\operatorname{start}}$ and the state at the end of the phase by $\gamma_u^{\operatorname{end}}$.
	Consider an arbitrary non-faulty node $v$ and the sets $A_v$ and $B_v$ in step (c).
	If $\gamma_v^{\operatorname{start}} = \gamma_v^{\operatorname{end}}$, then the claim is trivially true.
	So suppose $v \in B_v$ and $v$ receives identical value along $f+1$ node-disjoint $A_v v$-paths that exclude $F$ in step (a).
	Since the number of faulty nodes is at most $f$, thus at least one of these paths is both fault-free and has a non-faulty endpoint (other than $v$), say $u$.
	By Observation \ref{observation fault-free reliable}, it follows that whatever value is received by $v$ along this path in step (a) is the value flooded by $u$.
	Therefore, $\gamma_v^{\operatorname{end}} = \gamma_u^{\operatorname{start}}$, where $u$ is a non-faulty node, as required.
\end{proof_of}\\

\begin{proof_of}{Lemma \ref{lemma agreement}}
	Fix a phase of the algorithm and the corresponding set $F$ such that all faulty nodes are contained in $F$.
	Let $Z$ be the set of nodes that flooded $0$ in step (a) of the phase and let $N$ be the set of nodes that flooded $1$ in step (a).
	We first show that for any non-faulty node $v$, $Z_v = Z$ and $N_v = N$.
	Consider an arbitrary node $w \in Z$ (resp. $w \in N$) that flooded $0$ (resp. $1$) in step (a) of this phase.
	Observe that $P_{wv}$ identified in step (b) of the phase excludes $F$ and is fault-free.
	Therefore, by Observation \ref{observation fault-free reliable}, $v$ receives $0$ (resp. $1$) along $P_{wv}$ and correctly puts $w$ in the set $Z_v$ (resp. $N_v$), as required.
	
	It follows that for any two non-faulty nodes $u$ and $v$, we have that $Z_u = Z_v = Z$ and $N_u = N_v = N$.
	Thus $A_u = A_v$ and $B_u = B_v$.
	Let $A := A_u$ and $B := B_u$.
	First note that $A$ is always non-empty as follows.
	If $B$ is empty, then $A = V$ is non-empty.
	If $B$ is non-empty, then let $w \in B$ be a node in $B$.
	By Lemma \ref{lemma propagates}, there exist $f+1$ node-disjoint $A w$-paths, which implies that $\abs{A} \ge f+1$.
	Now all nodes in $A$ flooded identical value in step (a), say $\alpha$.
	If $u \in A$, then $u$'s state is $\alpha$ at the beginning of the phase and stays unchanged in step (c).
	Therefore, at the end of the phase $\gamma_u = \alpha$.
	If $u \in B$, then observe that the $f+1$ node-disjoint $A u$-paths identified by $u$ in step (c) are all fault-free.
	By Observation \ref{observation fault-free reliable}, it follows that $u$ receives $\alpha$ identically along these $f+1$ paths and so, at the end of the phase, $\gamma_u = \alpha$.
	Similarly for $v$, we have that $\gamma_v = \alpha$, as required.
\end{proof_of}

%% file: efficientAlgorithm.tex
In this section we give an efficient algorithm when $G$ is $2f$-connected to prove Theorem \ref{theorem efficient algorithm}.
We start by defining reliable message transmission.

\begin{definition} \label{definition reliable receive}
	For two nodes $u, v \in V$, node $v$ reliably receives a message flooded by node $u$ if
	\begin{enumerate}[topsep=0pt, itemsep=0pt]
		\item $u = v$,
		\item $v$ is a neighbor of $u$, or
		\item $v$ receives the message identically on at least $f + 1$ node-disjoint $uv$-paths.
	\end{enumerate}
\end{definition}

Observe that if a node $u$ broadcasts a message $M$, then any node $v$ cannot reliably receive a message $M' \ne M$ from $u$.

We will say that a non-faulty node is a \emph{type A} node if it knows the identity of all $f$ faulty nodes.
Any non-faulty node that is not \emph{type A} node is said
to be a \emph{type B} node. 
Initially, all non-faulty nodes are type B nodes.
As the algorithm executes, some non-faulty nodes will discover the identity of all the faulty nodes and become type A nodes.



The algorithm proceeds in three phases.
In phase 1, each node $v$ floods its input value.
As before, it requires $v$ to first locally broadcast its input to its neighbors and then assist other nodes to flood their values by forwarding received messages from its neighbors.
Recall that flooding completes in $n$ synchronous rounds.


In phase 2, $v$ floods all the messages it hears from its neighbors in phase 1. Recall that under local broadcast model each node receives all the messages transmitted by its neighbors.
After $n$ synchronous time steps, flooding in phase 2 will complete.
These messages help identify faulty nodes if they exhibit faulty behavior (i.e. tamper with messages).
The details are in the proof of Lemma \ref{lemma faulty node can not hide}.

In phase 3, type B nodes reach consensus and inform the type A nodes on the decision.
As discussed above, type A nodes can identify fault-free paths to each type B node.
If no type B node exists, then type A nodes decide by themselves.

The formal procedure is given in Algorithm \ref{algorithm efficient} and a proof of correctness follows.

\begin{algorithm}[ht]
	\SetAlgoLined
	\SetKwFor{For}{For}{do}{end}
	\begin{enumerate}[label=Phase \arabic*:,rightmargin=\widthof{Phase}]
	    \item 
    	Flood input value. (The steps taken to achieve flooding are described in the text preceding Algorithm \ref{algorithm consensus} pseudo-code.)
    
    	\item
    	For each neighbor $u$ of node $v$, $v$ reports on the messages propagated by $u$ in phase 1.
    	After the flooding of the reports is complete, node $v$ attempts to discover the faulty nodes as below.
    	
    	For each node $w$ in $G$ such that $v$ reliably received value $b$ from $w$, do:\\
    	\hspace*{0.3in} For each node $u \ne w$, identify $2f$ node-disjoint $uw$ paths.\\
    	\hspace*{0.3in} For each such path $P$ and node $z \in P$, if $v$ reliably receives that $z$ forwarded\\
    	\hspace*{0.3in} $\bar{b} = 1 - b$ as $w$'s value in phase 1, and $z$ is the first such node in $P$, then $v$\\
    	\hspace*{0.3in} sets $z$ to be a faulty node.
    
    	\item
    	After phase 2, $v$ is either a type A node or a type B node.
    	If $v$ is a type B node, then it decides from the input values it received reliably, by taking the majority (in case of a tie, $0$ is chosen as the majority value), and then floods the decision value.
    	Otherwise $v$ is a type A node.
    	Note that a type A node knows the identity of all $f$ faulty nodes and can thus identify fault-free paths from each node.
    	So a type A node can get the input value of a non-faulty node $u$ by accepting the value it receives on a fault-free $uv$-path.
    	$v$ waits for a decision value from a non-faulty type B node.
    	If a decision value is received along a fault-free path from a non-faulty (type B) node,\\
    	\hspace*{0.3in} then $v$ decides on this value;\\
    	\hspace*{0.3in} else $v$ decides from the input values of all non-faulty nodes by taking the\\
    	\hspace*{0.3in} majority.
	\end{enumerate}
	
	\caption{
		Proposed efficient algorithm for Byzantine consensus under the local broadcast model when $G$ is $2f$-connected: Steps performed by node $v$ are shown here.
	}
	\label{algorithm efficient}
\end{algorithm}

%% file: proofsEfficient.tex
The first lemma states that an adversary cannot hide a message broadcast by a faulty node.
\begin{lemma} \label{lemma faulty node can not hide}
	A message sent by a faulty node is received reliably by every node.
\end{lemma}

\begin{proof}
	Let $u$ be an arbitrary faulty node and let $v$ be an arbitrary node in $G$.
	If $u = v$ or $v$ is a neighbor of $u$, then the claim is trivially true from Definition \ref{definition reliable receive}.
	Otherwise, there exist $2f$ node-disjoint paths from $u$ to $v$ as the graph is $2f$-connected.
	Since $u$ is faulty, only $f-1$ of these paths can have faulty nodes.
	Any message sent by $u$ is received identically by all its neighbors, and so $v$ will receive it identically along the remaining $f+1$ node-disjoint paths that are fault-free.
\end{proof}

Next, we want to show that all type B nodes know the same set of input values.
We need an intermediate lemma to show that.

\begin{lemma} \label{lemma identify faulty nodes}
	Let $v$ and $w$ be distinct nodes.
	In phase 1, if $v$ reliably receives an input value of some node $u$ and $w$ does not, then $v$ knows the identity of $f$ faulty nodes after phase $2$.
\end{lemma}

\begin{proof}
	Fix a node $u$ with input $b \in \set{ 0, 1 }$ such that $v$ reliably received $b$ from $u$ in phase 1 but $w$ did not.
	Observe that, by Lemma \ref{lemma faulty node can not hide}, $u$ is a non-faulty node since $w$ did not reliably receive $b$ from $u$ in phase 1.
	Let $P_1, \dots, P_{2f}$ be $2f$ node-disjoint $uw$-paths.
	Since $w$ did not reliably receive $b$ from $u$ in phase 1, therefore exactly $f$ of these paths have faulty nodes.
	WLOG let these paths be $P_1, \dots, P_f$.
	By Lemma \ref{lemma faulty node can not hide}, after phase 2, $v$ reliably receives that some nodes on $P_1 \dots, P_f$ forwarded $\bar{b} = 1 - b$ in phase 1.
	For each path, $v$ sets the first such node to be faulty.

	To see why this assignment of faulty nodes by $v$ is correct, consider an arbitrary path in $P_1, \dots, P_f$.
	WLOG let this path be $P_1$.
	Let $z$ be the faulty node in $P_1$ that tampers the message.
	Observe that each of $P_1, \dots, P_f$ has exactly one faulty node that tampers the message.
	In phase 2, by Lemma \ref{lemma faulty node can not hide}, $v$ reliably receives that $z$ forwarded $\bar{b}$ in phase 1.
	Moreover, let $x$ be an arbitrary node in $P_1$ before $z$.
	Then $x$ is non-faulty and forwarded $b$ in phase 1.
	Therefore, in phase 2, $v$ can not reliably receive that $x$ forwarded $\bar{b}$ in phase 1.
\end{proof}

\begin{lemma} \label{lemma efficient agreement}
	All type B nodes reliably receive the same set of input values in phase 1.
\end{lemma}

\begin{proof}
	Suppose, for the sake of contradiction, that two nodes $v$ and $w$ are type B nodes and there exists a node $u$ such that $v$ reliably receives the input of $u$ and $w$ does not.
	Then, by Lemma \ref{lemma identify faulty nodes}, $v$ knows the identity of $f$ faulty nodes after phase 2.
	Therefore, $v$ must be a type A node, a contradiction.
\end{proof}

This next lemma shows that each node knows the input values of at least $2f$ other nodes.

\begin{lemma} \label{lemma efficient validity}
	Each node reliably receives input values of at least $2f$ other nodes in phase 1.
\end{lemma}

\begin{proof}
	Since the graph is $2f$ connected, therefore each node has at least $2f$ neighbors.
	By Definition \ref{definition reliable receive} each node reliably receives input from these $2f$ nodes in phase 1.
\end{proof}

We now have all the necessary lemmas to prove Theorem \ref{theorem efficient algorithm}.

\begin{proof_of}{Theorem \ref{theorem efficient algorithm}}
	Each phase of Algorithm \ref{algorithm efficient} requires at most $n$ synchronous rounds.
	Since there are only $3$ phases, so the algorithm terminates in $O(n)$ synchronous rounds.
	For validity and agreement, there are two cases to consider.

	\textbf{There is at least one type B node:}
	For agreement, note that by Lemma \ref{lemma efficient agreement} all type B nodes receive the same input values and therefore decide on the same value by taking the majority.
	A type A node (if one exists) knows the identity of all $f$ faulty nodes and so it can ignore messages on paths with faulty nodes.
	So we only need one fault-free path between a type A node and any type B node.
	Since there are at least $f$ node-disjoint fault free-paths between any two nodes, therefore in phase 3 each type A node correctly receives a message from a type B node about the final decision and decides on the same value.
	For validity, note that by Lemma \ref{lemma efficient validity} each type B node has input values of at least $2f + 1$ nodes (including its own).
	So the decided value is an input of at least $f+1$ nodes.
	Since there are at most $f$ faulty nodes, therefore the decided value is an input of at least one non-faulty node, as required.

	\textbf{There are no type B nodes:}
	Let $u$ be an arbitrary non-faulty node.
	Since there are no type B nodes, $u$ is a type A node that does not receive any decision value from any type B node in phase 3.
	Since $u$ knows the identity of $f$ faulty nodes, $u$ can identify fault-free paths to non-faulty nodes and receive untampered messages from non-faulty nodes in phase 1 and ignore messages from any path that contains a faulty node.
	Therefore, $u$ knows the input values of all non-faulty nodes.
	For agreement, observe that all non-faulty nodes are type A and each non-faulty node decides on the same decision value by taking the majority of the input values of non-faulty nodes.
	For validity, all non-faulty nodes consider the input values of non-faulty nodes.
\end{proof_of}

%% file: proofsHybridNecessity.tex
\begin{lemma} \label{lemma hybrid degree}
	If there exists a Byzantine consensus algorithm under the hybrid model on a graph $G$ tolerating at most $f$ Byzantine faulty nodes, of which at most $t$ nodes, $0 < t \le f$, are equivocating faulty nodes,
	then every set of nodes $S$, such that $0 < |S| \leq t$, has at least $2f + 1$ neighbors.
\end{lemma}

\begin{proof}
    Consider $t$ such that $0 < t \le f$.
	Let $\varphi = f - t$.
	For the sake of contradiction, suppose that there exists a non-empty set of nodes $S$ in $G$ of size at most $t$ that has at most $2f$ neighbors and there exists an algorithm $\mathcal{A}$ that solves Byzantine consensus under the hybrid model on $G$.
	Recall that neighbors of $S$ are nodes outside of $S$ that have an edge to some node in $S$.
	It is easy to show that when $n > f \ge t$, each set of cardinality at most $t$ must have at least one neighbor to achieve consensus.
	So in the rest of the proof we assume that $S$ has at least one neighbor.
	Let $N$ be the non-empty neighborhood of $S$.
	Then there exists a partition $(F^1, F^2, R, T)$ of $N$ such that $\abs{F^1}, \abs{F^2} \le \varphi$ and $\abs{R}, \abs{T} \le t$.
	Let $W = V - (S \cup N)$ be the set of remaining nodes.
	Note that some of these sets (other than $S$) can be possibly empty.
	However, since $n > t$ and $N$ is non-empty, we select these sets so that $R$ is necessarily non-empty.
	Recall that $\mathcal{A}$ outlines a procedure $\mathcal{A}_u$ for each node $u$ that describes $u$'s state transitions in each round.
	
	Similar to the proof of Lemma \ref{lemma degree}, we first create a network $\mathcal{G}$ to model behavior of nodes in $G$ in three different executions $E_1$, $E_2$, and $E_3$, which we will describe later.
	$\mathcal{G}$ consists of two copies of each node in $W$ and $T$, and a single copy of the remaining nodes.
	Figure \ref{figure hybrid degree network} depicts $\mathcal{G}$.
	The edges in $\mathcal{G}$ can be deduced from Figure \ref{figure hybrid degree network} along the lines of the proof of Lemma \ref{lemma degree}.
	As in proof of Lemma \ref{lemma degree}, the structure of $\mathcal{G}$ ensures that, for each edge $uv$ in the original graph $G$, each copy of $u$ receives messages from exactly one copy of $v$ in $\mathcal{G}$.
	This allows us to create an algorithm for $\mathcal{G}$ corresponding to $\mathcal{A}$ as follows.
	For each node $u \in G$, if $\mathcal{G}$ has one copy of $u$ then $u$ runs $\mathcal{A}_u$.
	Otherwise there are two copies $u_0$ and $u_1$ of $u$.
	Both $u_0$ and $u_1$ run $\mathcal{A}_u$.
	
	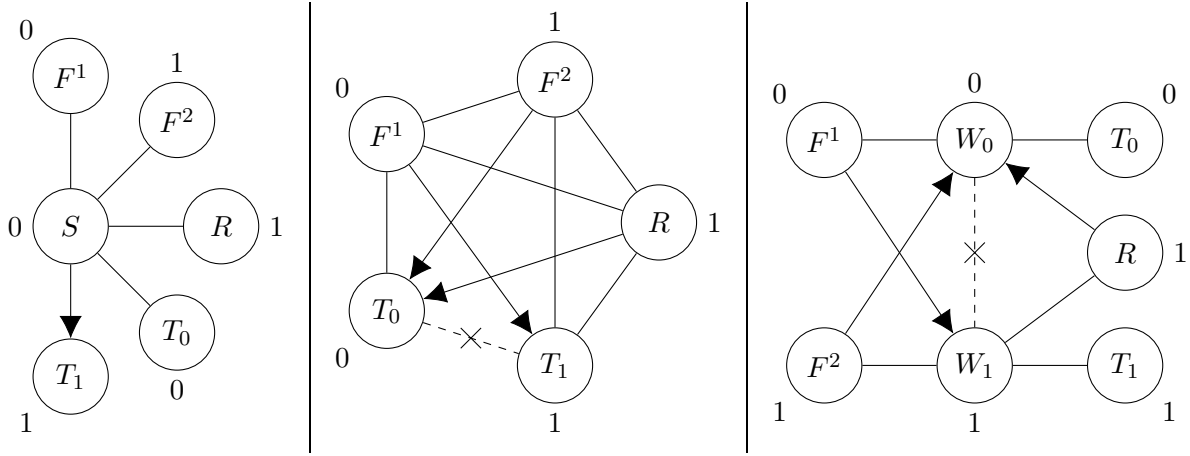
\begin{figure}[t]
		\centering
		\begin{tabular}{c|c|c}
				\begin{tikzpicture}
					\node[draw, circle, minimum height=1cm, label={left:$0$}]
						at (0, 0) (S) {$S$};
					\node[draw, circle, minimum height=1cm, label={above left:$0$}]
						at (90:2) (F^1) {$F^1$};
					\node[draw, circle, minimum height=1cm, label={above:$1$}]
						at (45:2) (F^2) {$F^2$};
					\node[draw, circle, minimum height=1cm, label={right:$1$}]
						at (0:2) (T_1) {$R$};
					\node[draw, circle, minimum height=1cm, label={below:$0$}]
						at (-45:2) (T_2) {$T_0$};
					\node[draw, circle, minimum height=1cm, label={below left:$1$}]
						at (-90:2) (T_2') {$T_1$};
					
					\draw[-] (S) to (F^1);
					\draw[-] (S) to (F^2);
					\draw[-] (S) to (T_1);
					\draw[-] (S) to (T_2);
					\draw[-{Latex[width=3mm,length=3mm]}] (S) to (T_2');
				\end{tikzpicture}
				&
				\begin{tikzpicture}
					\node[draw, circle, minimum height=1cm, label={above left:$0$}]
						at (144:2) (F^1) {$F^1$};
					\node[draw, circle, minimum height=1cm, label={above:$1$}]
						at (72:2) (F^2) {$F^2$};
					\node[draw, circle, minimum height=1cm, label={right:$1$}]
						at (0:2) (T_1) {$R$};
					\node[draw, circle, minimum height=1cm, label={below:$1$}]
						at (-72:2) (T_2') {$T_1$};
						\node[draw, circle, minimum height=1cm, label={below left:$0$}]
						at (-144:2) (T_2) {$T_0$};
					
					\draw[-] (F^1) to (F^2);
					\draw[-] (F^1) to (T_1);
					\draw[-] (F^1) to (T_2);
					\draw[-] (F^2) to (T_1);
					\draw[-] (F^2) to (T_2');
					\draw[-] (T_1) to (T_2');
					\draw[-{Latex[width=3mm,length=3mm]}] (F^1) to (T_2');
					\draw[-{Latex[width=3mm,length=3mm]}] (F^2) to (T_2);
					\draw[-{Latex[width=3mm,length=3mm]}] (T_1) to (T_2);
					\draw[dashed] (T_2) to node[cross out,draw,solid]{} (T_2');
				\end{tikzpicture}
				&
				\begin{tikzpicture}
					\node[draw, circle, minimum height=1cm, label={above left:$0$}]
						at (-2, 1.5) (F^1) {$F^1$};
					\node[draw, circle, minimum height=1cm, label={below left:$1$}]
						at (-2, -1.5) (F^2) {$F^2$};
					
					\node[draw, circle, minimum size=1cm, label={above:$0$}]
						at (0, 1.5) (W_0) {$W_0$};
					\node[draw, circle, minimum size=1cm, label={below:$1$}]
						at (0, -1.5) (W_1) {$W_1$};
					
					\node[draw, circle, minimum height=1cm, label={right:$1$}]
						at (2,0) (T_1) {$R$};
					\node[draw, circle, minimum height=1cm, label={above right:$0$}]
						at (2,1.5) (T_2) {$T_0$};
					\node[draw, circle, minimum height=1cm, label={below right:$1$}]
						at (2,-1.5) (T_2') {$T_1$};
					
					\draw[dashed] (W_0) to node[cross out,draw,solid]{} (W_1);
									
					\draw[-{Latex[width=3mm,length=3mm]}] (F^1) to (W_1);
					\draw[-] (F^1) to (W_0);
					
					\draw[-{Latex[width=3mm,length=3mm]}] (F^2) to (W_0);
					\draw[-] (F^2) to (W_1);
					
					\draw[-{Latex[width=3mm,length=3mm]}] (T_1) to (W_0);
					\draw[-] (T_1) to (W_1);
					\draw[-] (T_2') to (W_1);
					\draw[-] (T_2) to (W_0);
				\end{tikzpicture}
		\end{tabular}
		\caption{
			Network $\mathcal{G}$ to model executions $E_1$, $E_2$, and $E_3$.
			To reduce clutter, the edges are shown in three different figures.
			The left figure shows edges between $S$ and the neighborhood of $S$.
			The center figure shows edges within the neighborhood of $S$.
			The right figure shows edges between $W_0 \cup W_1$ and the neighborhood of $S$.
			Edges within the sets are not shown while edges between sets are depicted as single edges.
			The crossed dotted line between $W_0$ and $W_1$ (resp. $T_0$ and $T_1$) emphasizes that there are no edges between $W_0$ and $W_1$ (resp. $T_0$ and $T_1$).
			The numbers adjacent to the sets are the corresponding inputs in execution $\mathcal{E}$.
		}
		\label{figure hybrid degree network}
	\end{figure}
	
	Consider an execution $\mathcal{E}$ of the above algorithm on $\mathcal{G}$ as follows.
	Each node in $W_0 \cup T_0 \cup F^1 \cup S$ has input $0$ and the remaining nodes have input $1$.
	As in proof of Lemma \ref{lemma degree}, we will show that the algorithm does indeed terminate but nodes do not reach agreement in $\mathcal{G}$, which will be useful in deriving the desired contradiction.
	We use $\mathcal{E}$ to describe three executions $E_1$, $E_2$, and $E_3$ of $\mathcal{A}$ on the original graph $G$ as follows.
	\begin{enumerate}[label=$E_{\arabic*}$:,topsep=0pt,labelindent=0pt,itemsep=0pt]
		\item
		$F^2 \cup R$ is the set of faulty nodes and they are all non-equivocating faulty nodes.
		In each round, a faulty node in $F^2 \cup R$ broadcasts the same messages as the corresponding node in $\mathcal{G}$ in execution $\mathcal{E}$ in the same round.
		All non-faulty nodes have input $0$.
		Note that the behavior of non-faulty nodes in $S \cup F^1$, $T$, and $W$ is modelled by the corresponding (copies of) nodes in $S \cup F^1$, $T_0$, and $W_0$, respectively, in $\mathcal{E}$.
        Since $\mathcal{A}$ solves Byzantine consensus on $G$, nodes in $S \cup F^1 \cup T \cup W$ decide on output $0$ (by validity) in finite time.
		
		\item
		$F^1$ is the set of non-equivocating faulty nodes and $T$ is the set of equivocating faulty nodes.
		In each round, a non-equivocating faulty node in $F^1$ broadcasts the same messages as the corresponding node in $\mathcal{G}$ in execution $\mathcal{E}$ in the same round.
		Recall that a equivocating faulty node in $T$ has the capability of point-to-point communication with its neighbors.
		The communication by equivocating faulty nodes in $T$ to its neighbors in $S$ is the same as that by the corresponding copy in $T_0$ and to the remaining non-faulty neighbors is the same as that by the corresponding copy in $T_1$ in execution $\mathcal{E}$.
		$S$ has input $0$ and all the remaining non-faulty nodes have input $1$.
		Note that the behavior of non-faulty nodes in $S \cup F^2$, $R$ and $W$ is modelled by the corresponding (copies of) nodes in $S \cup F^2$, $R$, and $W_1$, respectively, in $\mathcal{E}$.
        The output of the non-faulty nodes will be described later.
		
		\item
		$F^1 \cup S$ is the set of faulty nodes and they are all non-equivocating faulty nodes.
		In each round, a faulty node in $F^1 \cup S$ broadcasts the same messages as the corresponding node in $\mathcal{G}$ in execution $\mathcal{E}$ in the same round.
		All non-faulty nodes have input $1$.
		Note that the behavior of non-faulty nodes in $R \cup F^2$, $T$, and $W$ is modelled by the corresponding (copies of) nodes in $R \cup F^2$, $T_1$, and $W_1$, respectively, in $\mathcal{E}$.
        Since $\mathcal{A}$ solves Byzantine consensus on $G$, nodes in $F^2 \cup R \cup T \cup W$ decide on output $1$ (by validity) in finite time.
	\end{enumerate}
	
    Due to the output of nodes in $S$ in $E_1$, the nodes in $S$ output $0$ in $\mathcal{E}$.
    Similarly, due to the output of nodes in $R$ in $E_3$, the nodes in $R$ output $1$ in $\mathcal{E}$.
	It follows that in $E_2$, nodes in $R$ as modeled by $R$ in $\mathcal{E}$ output $1$ while $S$ outputs $0$.
	Recall that, by construction, both $S$ and $R$ are non-empty.
	This violates agreement, a contradiction.
\end{proof}

\begin{lemma} \label{lemma hybrid connectivity}
	If there exists a Byzantine consensus algorithm under the hybrid model on a graph $G$ tolerating at most $f$ Byzantine faulty nodes, of which at most $t$ are equivocating faulty nodes, then $G$ is $(\floor{3(f-t)/2} + 2t + 1)$-connected.
\end{lemma}

\begin{proof}
	Suppose for the sake of contradiction that $G$ is not $(\floor{3(f-t)/2} + 2t + 1)$-connected and there exists an algorithm $\mathcal{A}$ that solves Byzantine consensus under the hybrid model on $G$.
	Let $\varphi = f - t$.
	Let $C$ be a vertex cut of $G$ of size at most $\floor{ 3\varphi/2 } + 2t$ with a partition $(A, B, C)$ of $V$ such that $A$ and $B$ (both non-empty) are disconnected in $G - C$ (so there is no edge between a node in $A$ and a node in $B$).
	Since $\abs{C} \le \floor{3\varphi/2} + 2t$, there exists a partition $(C^1, C^2, C^3, R, T)$ of $C$ such that $\abs{C^1}, \abs{C^2} \le \floor{\varphi/2}$, $\abs{C^3} \le \ceil{\varphi/2}$, and $\abs{R}, \abs{T} \le t$.
	Recall that $\mathcal{A}$ outlines a procedure $\mathcal{A}_u$ for each node $u$ that describes $u$'s state transitions in each round.
	
	Similar to the proof of Lemma \ref{lemma degree}, we first create a network $\mathcal{G}$ to model behavior of nodes in $G$ in three different executions $E_1$, $E_2$, and $E_3$, which we will describe later.
	$\mathcal{G}$ consists of two copies of each node in $A$, $B$, $R$, and $T$, and a single copy of the remaining nodes.
	Figure \ref{figure hybrid connectivity network} depicts $\mathcal{G}$.
	The edges in $\mathcal{G}$ can be deduced from Figure \ref{figure hybrid connectivity network} along the lines of the proof of Lemma \ref{lemma degree}.
	As in proof of Lemma \ref{lemma degree}, the structure of $\mathcal{G}$ ensures that, for each edge $uv$ in the original graph $G$, each copy of $u$ receives messages from exactly one copy of $v$ in $\mathcal{G}$.
	This allows us to create an algorithm for $\mathcal{G}$ corresponding to $\mathcal{A}$ as follows.
	For each node $u \in G$, if $\mathcal{G}$ has one copy of $u$ then $u$ runs $\mathcal{A}_u$.
	Otherwise there are two copies $u_0$ and $u_1$ of $u$.
	Both $u_0$ and $u_1$ run $\mathcal{A}_u$.
	
	\begin{figure}[t]
		\centering
		\begin{tabular}{c|c|c}
				\begin{tikzpicture}
					\node[draw, circle, minimum size=1cm, label={above:$0$}] at (0, 0) (A_0) {$A_0$};
					\node[draw, circle, minimum size=1cm, label={below:$1$}] at (0, -3) (A_1) {$A_1$};
					\node[draw, ellipse, minimum height=1cm, label={above:$0$}] at (2, 1) (C^1) {$C^1$};
					\node[draw, ellipse, minimum height=1cm, label={above:$1$}] at (2, -2) (C^2) {$C^2$};
					\node[draw, ellipse, minimum height=1cm, label={below:$1$}] at (2, -5) (C^3) {$C^3$};
					\node[draw, circle, minimum height=1cm, label={above:$0$}] at (-2, 2) (R_0) {$R_0$};
					\node[draw, circle, minimum height=1cm, label={above:$0$}] at (-2, -0.5) (T_0) {$T_0$};
					\node[draw, circle, minimum height=1cm, label={below:$1$}] at (-2, -3) (T_1) {$T_1$};
					\node[draw, circle, minimum height=1cm, label={below:$1$}] at (-2, -5.5) (R_1) {$R_1$};
					
					\draw[dashed] (A_0) to node[cross out,draw,solid]{} (A_1);
					
					\draw[-{Latex[width=3mm,length=3mm]}] (C^1) to (A_1);
					\draw[-{Latex[width=3mm,length=3mm]}] (C^2) to (A_1);
					\draw[-{Latex[width=3mm,length=3mm]}] (C^3) to (A_0);
					
					\draw[-] (C^1) to (A_0);
					\draw[-] (C^2) to (A_0);
					\draw[-] (C^3) to (A_1);
					
					\draw[-] (R_0) to (A_0);
					\draw[-] (T_0) to (A_1);
					\draw[-] (R_1) to (A_1);
					\draw[-] (T_1) to (A_0);
				\end{tikzpicture}
				&
				\begin{tikzpicture}
					\node[draw, circle, minimum height=1cm, label={above:$0$}]
						at (45:2.5) (C^1) {$C^1$};
					\node[draw, circle, minimum height=1cm, label={above:$1$}]
						at (0:2.5) (C^2) {$C^2$};
					\node[draw, circle, minimum height=1cm, label={below:$1$}]
						at (-45:2.5) (C^3) {$C^3$};
					\node[draw, circle, minimum height=1cm, label={above:$0$}]
						at (99:2.5) (R_0) {$R_0$};
					\node[draw, circle, minimum height=1cm, label={above:$1$}]
						at (153:2.5) (R_1) {$R_1$};
					\node[draw, circle, minimum height=1cm, label={below:$1$}]
						at (207:2.5) (T_1) {$T_1$};
					\node[draw, circle, minimum height=1cm, label={below:$0$}]
						at (261:2.5) (T_0) {$T_0$};
					
					\draw[-] (C^1) to (C^2);
					\draw[-] (C^1) to (C^3);
					\draw[-] (C^2) to (C^3);
					
					\draw[-] (R_0) to (T_0);
					\draw[-] (R_1) to (T_1);
					\draw[dashed] (R_0) to node[cross out,draw,solid]{} (R_1);
					\draw[dashed] (T_0) to node[cross out,draw,solid]{} (T_1);
					
					\draw[-] (C^1) to (R_0);
					\draw[-] (C^1) to (T_0);
					\draw[-] (C^2) to (R_1);
					\draw[-] (C^2) to (T_1);
					\draw[-] (C^3) to (R_1);
					\draw[-] (C^3) to (T_0);
					
					\draw[-{Latex[width=3mm,length=3mm]}] (C^1) to (R_1);
					\draw[-{Latex[width=3mm,length=3mm]}] (C^1) to (T_1);
					\draw[-{Latex[width=3mm,length=3mm]}] (C^2) to (R_0);
					\draw[-{Latex[width=3mm,length=3mm]}] (C^2) to (T_0);
					\draw[-{Latex[width=3mm,length=3mm]}] (C^3) to (R_0);
					\draw[-{Latex[width=3mm,length=3mm]}] (C^3) to (T_1);
				\end{tikzpicture}
				&
				\begin{tikzpicture}
					\node[draw, circle, minimum size=1cm, label={above:$0$}] at (0, 0) (B_0) {$B_0$};
					\node[draw, circle, minimum size=1cm, label={below:$1$}] at (0, -3) (B_1) {$B_1$};
					\node[draw, ellipse, minimum height=1cm, label={above:$0$}] at (-2, 1) (C^1) {$C^1$};
					\node[draw, ellipse, minimum height=1cm, label={above:$1$}] at (-2, -2) (C^2) {$C^2$};
					\node[draw, ellipse, minimum height=1cm, label={below:$1$}] at (-2, -5) (C^3) {$C^3$};
					\node[draw, circle, minimum height=1cm, label={above:$0$}] at (2, 2) (R_0) {$R_0$};
					\node[draw, circle, minimum height=1cm, label={above:$0$}] at (2, -0.5) (T_0) {$T_0$};
					\node[draw, circle, minimum height=1cm, label={below:$1$}] at (2, -3) (T_1) {$T_1$};
					\node[draw, circle, minimum height=1cm, label={below:$1$}] at (2, -5.5) (R_1) {$R_1$};
					
					\draw[dashed] (B_0) to node[cross out,draw,solid]{} (B_1);
					
					\draw[-{Latex[width=3mm,length=3mm]}] (C^1) to (B_1);
					\draw[-{Latex[width=3mm,length=3mm]}] (C^2) to (B_0);
					\draw[-{Latex[width=3mm,length=3mm]}] (C^3) to (B_0);
					
					\draw[-] (C^1) to (B_0);
					\draw[-] (C^2) to (B_1);
					\draw[-] (C^3) to (B_1);
					
					\draw[-] (R_0) to (B_0);
					\draw[-] (T_0) to (B_0);
					\draw[-] (R_1) to (B_1);
					\draw[-] (T_1) to (B_1);
				\end{tikzpicture}
		\end{tabular}
		\caption{
			Network $\mathcal{G}$ to model executions $E_1$, $E_2$, and $E_3$.
			To reduce clutter, the edges are shown in three different figures.
			The left figure shows edges between the copies of nodes in $A$ and the copies of nodes in the cut $C$.
			The center figure shows edges among the copies of nodes in $C$.
			The right figure shows edges between the copies of nodes in $A$ and the copies of nodes in the cut $C$.
			Edges within the sets are not shown while edges between sets are depicted as single edges.
			The crossed dotted lines emphasize that there are no edges between the corresponding sets.
			The numbers adjacent to the sets are the corresponding inputs in execution $\mathcal{E}$.
		}
		\label{figure hybrid connectivity network}
	\end{figure}
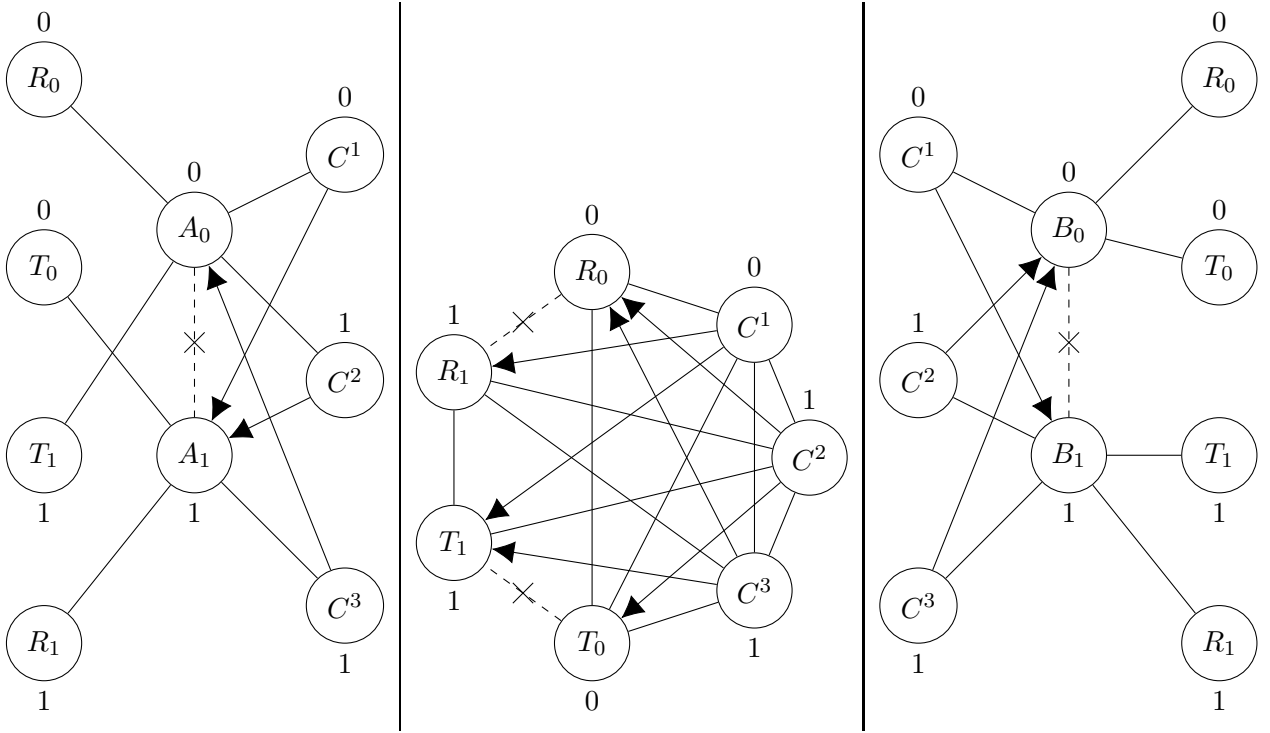
	
	Consider an execution $\mathcal{E}$ of the above algorithm on $\mathcal{G}$ as follows.
	Each node in $A_0$, $B_0$, $R_0$, $T_0$, and $C^1$ has input $0$ and the remaining nodes have input $1$.
	As in proof of Lemma \ref{lemma degree}, we will show that the algorithm does indeed terminate but nodes do not reach agreement in $\mathcal{G}$, which will be useful in deriving the desired contradiction.
	We use $\mathcal{E}$ to describe three executions $E_1$, $E_2$, and $E_3$ of $\mathcal{A}$ on the original graph $G$ as follows.
	\begin{enumerate}[label=$E_{\arabic*}$:,topsep=0pt,labelindent=0pt,itemsep=0pt]
		\item
		$C^2 \cup C^3$ is the set of non-equivocating faulty nodes and $T$ is the set of equivocating faulty nodes.
		In each round, a non-equivocating faulty node in $C^2 \cup C^3$ broadcasts the same messages as the corresponding node in $\mathcal{G}$ in execution $\mathcal{E}$ in the same round.
		Recall that a equivocating faulty node in $T$ has the capability of point-to-point communication with its neighbors.
		The communication by equivocating faulty nodes in $T$ to its neighbors in $A$ is the same as that by the corresponding copy in $T_1$ and to the remaining non-faulty neighbors is the same as that by the corresponding copy in $T_0$ in execution $\mathcal{E}$.
		All non-faulty nodes have input $0$.
		Note that the behavior of non-faulty nodes in $A$, $B$, $R$, and $C^1$ is modelled by the corresponding (copies of) nodes in $A_0$, $B_0$, $R_0$, and $C^1$, respectively, in $\mathcal{E}$.
        Since $\mathcal{A}$ solves Byzantine consensus on $G$, nodes in $A \cup B \cup R \cup C^1$ decide on output $0$ (by validity) in finite time.
		
		\item
		$C^1 \cup C^3$ is the set of non-equivocating faulty nodes and $R$ is the set of equivocating faulty nodes.
		In each round, a non-equivocating faulty node in $C^1 \cup C^3$ broadcasts the same messages as the corresponding node in $\mathcal{G}$ in execution $\mathcal{E}$ in the same round.
		The communication by equivocating faulty nodes in $R$ to its neighbors in $A$ is the same as that by the corresponding copy in $R_0$ and to the remaining non-faulty neighbors is the same as that by the corresponding copy in $R_1$ in execution $\mathcal{E}$.
		$A$ has input $0$ and all the remaining non-faulty nodes have input $1$.
		Note that the behavior of non-faulty nodes in $A$, $B$, $T$, and $C^2$ is modelled by the corresponding (copies of) nodes in $A_0$, $B_1$, $T_1$, and $C^2$, respectively, in $\mathcal{E}$.
        The output of the non-faulty nodes will be described later.
		
		\item
		$C^1 \cup C^2$ is the set of non-equivocating faulty nodes and $T$ is the set of equivocating faulty nodes.
		In each round, a non-equivocating faulty node in $C^1 \cup C^2$ broadcasts the same messages as the corresponding node in $\mathcal{G}$ in execution $\mathcal{E}$ in the same round.
		The communication by equivocating faulty nodes in $T$ to its neighbors in $B$ is the same as that by the corresponding copy in $T_1$ and to the remaining non-faulty neighbors is the same as that by the corresponding copy in $T_0$ in execution $\mathcal{E}$.
		All non-faulty nodes have input $1$.
		Note that the behavior of non-faulty nodes in $A$, $B$, $R$, and $C^3$ is modelled by the corresponding (copies of) nodes in $A_1$, $B_1$, $R_1$, and $C^3$, respectively, in $\mathcal{E}$.
        Since $\mathcal{A}$ solves Byzantine consensus on $G$, nodes in $A \cup B \cup R \cup C^3$ decide on output $1$ (by validity) in finite time.
	\end{enumerate}
	
    Due to the output of nodes in $A$ in $E_1$, the nodes in $A_0$ output $0$ in $\mathcal{E}$.
    Similarly, due to the output of nodes in $B$ in $E_3$, the nodes in $B_1$ output $1$ in $\mathcal{E}$.
	It follows that in $E_2$, nodes in $A$ as modeled by $A_0$ in $\mathcal{E}$ output $0$ while nodes in $B$ as modeled by $B_1$ output $1$.
	Recall that, by construction, both $A$ and $B$ are non-empty.
	This violates agreement, a contradiction.
\end{proof}

We now prove the necessity of the conditions in Theorem \ref{theorem hybrid}.
The sufficiency is proved in the next section.

\begin{proof_of}{necessity of the conditions in Theorem \ref{theorem hybrid}}
	Condition (i) follows from Lemma \ref{lemma hybrid connectivity}.
	For $t = 0$ condition (ii) follows from Lemma \ref{lemma degree} and for  $t > 0$ condition (iii) follows from Lemma \ref{lemma hybrid degree}.
\end{proof_of}

%% file: proofsHybridSufficiency.tex
In this section, we assume that the graph $G$ satisfies the conditions in Theorem \ref{theorem hybrid}.
To prove sufficiency of the conditions in Theorem \ref{theorem hybrid}, we modify Algorithm \ref{algorithm consensus} for the hybrid setting.
The formal procedure is given in Algorithm \ref{algorithm hybrid consensus}.

\begin{algorithm}[ht]
	\SetAlgoLined
	\SetKwFor{For}{For}{do}{end}
	Each node $v$ has a binary input value in $\set{ 0, 1 }$, and maintains a binary state $\gamma_v \in \set{ 0, 1 }$.\\
	{\em Initialization:} $\gamma_v$ := input value of node $v$\\
	\For{each pair of sets $F, T \subseteq V$ such that $F \subseteq V - T$, $\abs{T} \le t$, and $\abs{F} \le f - \abs{T}$}{
		\begin{enumerate}
			[label=Step (\alph*):,labelindent=12pt,leftmargin=!,rightmargin=\widthof{Step ()}]
			\item
			Flood value $\gamma_v$. (The steps taken to achieve flooding are described in the text preceding Algorithm \ref{algorithm consensus} pseudo-code.)
			
			\item
			For each node $u \in V - T$, identify a single $uv$-path $P_{uv}$ that excludes $F \cup T$.
			Let,
			\begin{align*}
    			Z_v &:= \set{ u \in V - T \mid \text{$v$ received value $0$ from $u$ along $P_{uv}$ in step (a)} },\\
    			N_v &:= V - T - Z_v.
			\end{align*}
			
			\item
			Let,
			\begin{align*}
			    \varphi = f - \abs{T}.
			\end{align*}
			Define sets $A_v$ and $B_v$ as follows.
			\begin{enumerate}
				[label=Case \arabic*:]
				\item If $\abs{Z_v \cap F} \le \floor{\varphi/2}$ and $\abs{N_v} > f$, then
				 $A_v := N_v$ and $B_v := Z_v$.
				
				\item If $\abs{Z_v \cap F} \le \floor{\varphi/2}$ and $\abs{N_v} \le f$,
				then $A_v := Z_v$ and $B_v := N_v$.
				
				\item If $\abs{Z_v \cap F} > \floor{\varphi/2}$ and $\abs{Z_v} > f$,
			    then $A_v := Z_v$ and $B_v := N_v$.
				
				\item If $\abs{Z_v \cap F} > \floor{\varphi/2}$ and $\abs{Z_v} \le f$,
				then $A_v := N_v$ and $B_v := Z_v$.
			\end{enumerate}
			
			If $v \in B_v$ and $v$ receives value $\delta\in\{0,1\}$ along any $f+1$ node-disjoint $A_v v$-paths that exclude $F \cup T$ in step (a), then $\gamma_v := \delta$.
			
	\end{enumerate}}
	
	Output $\gamma_v$.
	
	\caption{
		Proposed algorithm for Byzantine consensus under the hybrid model: Steps performed by node $v$ are shown here.
	}
	\label{algorithm hybrid consensus}
\end{algorithm}

For correctness, we follow the same ideas as the proof of Theorem \ref{theorem sufficiency}.
We assume that $t > 0$ in the proof since Algorithm \ref{algorithm hybrid consensus} is the same as Algorithm \ref{algorithm consensus} when $t = 0$.
We follow the same strategy as in Section \ref{section correctness}.
Let $T^*$ be the actual equivocating faulty set of size at most $t$ and $F^*$ the actual non-equivocating faulty set of size at most $f - \abs{T^*}$.

Similar to Observation \ref{observation fault-free reliable}, the following follows from the rules used for flooding.

\begin{observation} \label{observation hybrid fault-free reliable}
	For any phase of Algorithm \ref{algorithm hybrid consensus}, in step (a) for any two nodes $u, v \in V - T^*$ (possibly non-equivocating faulty), if $v$ receives value $b$ along a fault-free $uv$-path then $u$ broadcast the value $b$ to its neighbors during flooding.
\end{observation}

The next two lemmas state that the paths identified in steps (b) and (c), respectively, do indeed exist.

\begin{lemma}
    For any choice of sets $F$ and $T$ in Algorithm \ref{algorithm hybrid consensus}, and any two nodes $u, v$, there exists a $uv$-path that excludes $F \cup T$.
\end{lemma}

\begin{proof}
    Recall that a path is said to exclude $F \cup T$ if none of the \emph{internal} nodes in the path belong to $F \cup T$.
    Since $G$ is $(\floor{3(f-t)/2}+2t+1)$-connected, by Menger's Theorem, there are at least $(\floor{3(f-t)/2}+2t+1)$ node-disjoint paths between any two nodes $u, v$.
	For any $f, t$ such that $0 \le t \le f$, we have that $\floor{3(f-t)/2}+2t+1 \ge (f - t) + t + 1 = f + 1$.
	Thus, there are at least $f+1$ node-disjoint $uv$-paths, of which at least one path must exclude $F \cup T$, since $\abs{F \cup T} \le f$.
\end{proof}

\begin{lemma} \label{lemma hybrid propagates}
	For any non-faulty node $v$ and any given phase of Algorithm \ref{algorithm hybrid consensus} with the corresponding sets $F$ and $T$, in step (c), if $v \in B_v$, then there exist $f+1$ node-disjoint $A_v v$-paths that exclude $F \cup T$.
\end{lemma}

\begin{proof}
	Fix a phase in the algorithm and the corresponding sets $F$ and $T$.
	Let $\varphi = f - \abs{T}$ as in the algorithm.
	Consider an arbitrary non-faulty node $v$ such that $v \in B_v$ in step (c).
	There are 4 cases to consider, corresponding to the 4 cases in step (c).
	Note that neither $A_v$ nor $B_v$ contain any node from $T$.
	\begin{enumerate}
		[label=Case \arabic*:,topsep=0pt,labelindent=0pt]
		\item $\abs{Z_v \cap F} \le \floor{\varphi/2}$ and $\abs{N_v} > f$.
		Then $A_v := N_v$ and $B_v := Z_v$.
		Therefore there exist at least $f+1$ nodes in $A_v$.
		Node $v$ selects $f+1$ nodes $A'_v$ from $A_v$ by choosing all nodes from $A_v \cap F$ and the rest arbitrarily from $A_v - F$.
        Define $B'_v = B_v \cap (F - v)$.
        Now, $\abs{B'_v} \le \abs{B_v \cap F} = \abs{Z_v \cap F} \le \floor{\varphi/2}$ by assumption of case 1.
		Since $G$ is $(\floor{3(f-t)/2}+2t+1)$-connected, the vertex connectivity of $G - T - B'_v$ is at least
		\begin{align*}
		    \floor{ \frac{3(f-t)}{2} } + 2t + 1 - \abs{T} - \floor{ \frac{\varphi}{2}}
		        &=  f - t + \floor{ \frac{f-t}{2} } + 2t + 1 - \abs{T} - \floor{ \frac{ f - \abs{T} }{2}}    \\
		        &=  1 + t + \floor{ \frac{f-t}{2} } + f - \abs{T} - \floor{ \frac{ f - \abs{T} }{2}}    \\
		        &=  1 + t + \floor{ \frac{f-t}{2} } + \ceil{ \frac{ f - \abs{T} }{2}}    \\
		        &\ge    1 + t + \floor{ \frac{f-t}{2} } + \ceil{ \frac{ f - t }{2}}    \\
		        &=  f - t + t + 1   \\
		        &=  f + 1.
		\end{align*}
		So there exist $f+1$ node-disjoint $A'_v v$-paths in $G-T-B'_v$.
		Furthermore, since all the nodes in $A_v \cap F$ are endpoints in these paths and $F = (A_v \cap F) \cup (B_v \cap F)$, we have that these paths exclude\footnote{Recall that a path is said to exclude $X$ if none of the {\em internal} nodes in the path belong to $X$.} $F \cup T$.
		
		\item $\abs{Z_v \cap F} \le \floor{\varphi/2}$ and $\abs{N_v} \le f$.
		Then $A_v := Z_v$ and $B_v := N_v$.
		Recall that $0 < t \le f$ and $\abs{F} \le \varphi$.
		Let $F'$ be a maximal set such that
		\begin{enumerate}[label=(\roman*)]
			\item $F' \subseteq B_v - v$,
			\item $F' \supseteq (F \cap B_v) - v$, and
			\item $\abs{F'} \le \varphi$,
		\end{enumerate}
		i.e., there is no set $F'' \supseteq F'$ such that
		\begin{enumerate}[label=(\roman*)]
			\item $F'' \subseteq B_v - v$,
			\item $F'' \supseteq (F \cap B_v) - v$, and
			\item $\abs{F''} \le \varphi$.
		\end{enumerate}
		Note that $F'$ exists because $\abs{F} \le \varphi$.
		Let $B'_v = B_v - F'$.
		We first show that $\abs{B'_v} \le t$.
		There are two cases to consider.
		In the first case, $\abs{F'} < \varphi$.
		Then $v$ is the only node in $B'_v = B_v - F'$ as follows.
		Suppose there is another node $u \in B'_v = B_v - F'$.
		Then $F' \cup \set{u} \subseteq B_v - v$ and $F' \cup \set{u} \supseteq F' \supseteq (F \cap B_v) - v$.
		Moreover, $\abs{F' \cup \set{u}} = \abs{F'} + 1 \le \varphi$, a contradiction since $F'$ is maximal.
		So we have that $\abs{B_v'} = 1 \le t$.
		In the second case, $\abs{F'} = \varphi$, and we have that $\abs{B_v'} = \abs{B_v} - \varphi = \abs{N_v} - \varphi \le f - \varphi = \abs{T} \le t$.
		
		We construct a graph $H$ from $G - T - F'$ as follows.
		Remove all edges $uw$ such that both $u, w \in A_v$.
		Add a node $r$ with edges to all nodes in $A_v$.
		We first show that $H$ has $f+1$ node-disjoint $rv$-paths.
		So suppose, for the sake of contradiction, that $H$ does not have $f+1$ node-disjoint $rv$-paths.
		Then by Menger's Theorem, there exists an $rv$-cut $C$ of size at most $f$.
		This partitions $H - C$ into $(R, S)$ such that $R$ and $S$ are both non-empty and disconnected with $r \in R$ and $v \in S$.
		Since $r$ has an edge to each node in $A_v$, we have that $S \subseteq B_v'$ and so $\abs{S} \le t$.
		All the neighbors of $S$ in $H$ are contained entirely in $C$\footnote{Recall that neighbors of $S$ are nodes outside of $S$ that have an edge to some node in $S$.}.
		Since no edge incident to any node in $B_v'$ in $G$ was removed when constructing graph $H$, so all the neighbors of $S$ in $G$ are contained entirely in $C \cup T \cup F'$.
		It follows that $S$ has at most $f + \abs{T} + \varphi = f + \abs{T} + f - \abs{T} = 2f$ neighbors in $G$, a contradiction with the condition (iii) in Theorem \ref{theorem hybrid}.
		
		We now show that $f+1$ node-disjoint $rv$-paths in $H$ imply that there are $f+1$ node-disjoint $A_v v$-paths in $G$ that exclude $F \cup T$.
		Observe that an $rv$-path $P$ in $H$ can have more than one node from $A_v$.
		However, since $r$ has an edge to each node in $A_v$, a shorter sub-path $P'$ exists that contains only the last node from $A_v$ in $P$.
		Furthermore, if a path $Q$ and $P$ are node-disjoint, then so are $Q$ and $P'$.
		So there exist $f+1$ node-disjoint $rv$-paths that each contain exactly one node from $A_v$.
		By removing the endpoint $r$ from each of these paths, we get $f+1$ node-disjoint $A_v v$-paths.
		Since each of these paths have exactly one node from $A_v$, this must necessarily be an endpoint.
		So each of these paths exclude $A_v$.
		By construction of $H$, these paths exist in $G - T - F'$ and so exclude $F'$ and $T$ as well.
						
		Now, $A_v \supseteq A_v \cap F$ and $F = (A_v \cap F) \cup (B_v \cap F)$.
		Recall that $F' \supseteq (B_v \cap F) - v$.
		It follows that the $f+1$ node-disjoint $A_v v$-paths above exclude $A_v \cup F' \cup T \supseteq (F - v) \cup T$.
		Since $v$ can only be an endpoint in these paths, we have that these $f+1$ node-disjoint $A_v v$-paths exclude $F \cup T$, as required.
		
		\item $\abs{Z_v \cap F} > \floor{\varphi/2}$ and $\abs{Z_v} > f$.
		Then $A_v := Z_v$ and $B_v := N_v$.
		We have that
		\begin{align*}
		\abs{N_v \cap F}
		&=      \abs{F} - \abs{Z_v \cap F}  \\
		&\le    \varphi - \floor{\varphi/2} - 1 \\
		&\le    \floor{\varphi/2}.
		\end{align*}
		So this case is the same as Case 1 with the roles of $Z_v$ and $N_v$ swapped.
		
		\item $\abs{Z_v \cap F} > \floor{\varphi/2}$ and $\abs{Z_v} \le f$.
		Then $A_v := N_v$ and $B_v := Z_v$.
		From the analysis in Case 3, we have that $\abs{N_v \cap F} \le \floor{\varphi/2}$.
		So this case is the same as Case 2 with the roles of $Z_v$ and $N_v$ swapped.
	\end{enumerate}
	So we have that there do exist $f+1$ node-disjoint $A_v v$-paths that exclude $F \cup T$.
\end{proof}

The next lemma states that in any phase, the state of a non-faulty node at the end of the phase equals the state of some non-faulty node at the beginning of the phase.

\begin{lemma} \label{lemma hybrid validity}
	For any non-faulty node $v$, its state $\gamma_v$ at the end of any given phase equals the state of some non-faulty node at the start of that phase.
\end{lemma}

\begin{proof}
	Fix a phase in the algorithm and the corresponding sets $F$ and $T$.
	For any node $u$, we denote the state at the beginning of the phase by $\gamma_u^{\operatorname{start}}$ and the state at the end of the phase by $\gamma_u^{\operatorname{end}}$.
	Consider an arbitrary non-faulty node $v$ and the sets $A_v$ and $B_v$ in step (c).
	If $\gamma_v^{\operatorname{start}} = \gamma_v^{\operatorname{end}}$, then the claim is trivially true.
	So suppose $v \in B_v$ and $v$ receives identical value along $f+1$ node-disjoint $A_v v$-paths in step (a).
	Since the number of faulty nodes is at most $f$, at least one of these paths is both fault-free and has a non-faulty endpoint (other than $v$), say $u$.
	By Observation \ref{observation hybrid fault-free reliable}, it follows that whatever value is received by $v$ along this path in step (a) is the value flooded by $u$.
	Therefore, $\gamma_v^{\operatorname{end}} = \gamma_u^{\operatorname{start}}$, where $u$ is a non-faulty node, as required.
\end{proof}

Next, we show that when the sets $F$ and $T$ are properly selected, all non-faulty nodes reach agreement in that phase.

\begin{lemma} \label{lemma hybrid agreement}
	Consider the unique phase of the Algorithm \ref{algorithm hybrid consensus} where $F = F^*$ and $T = T^*$.
	At the end of this phase, every pair of non-faulty nodes $u, v \in V$ have identical state, i.e., $\gamma_u = \gamma_v$.
\end{lemma}

\begin{proof}
	Consider the unique phase of the algorithm where $F = F^*$ and $T = T^*$.
	By choice of $T$, all equivocating faulty nodes are in $T$ so each node in $V - T$ floods one value, and one value only, in step (a) of the phase.
	Let $Z \subseteq V - T$ be the set of nodes that flooded $0$ in step (a) and let $N \subseteq V - T$ be the set of nodes that flooded $1$ in step (a).
	We first show that for any non-faulty node $v$, $Z_v = Z$ and $N_v = N$.
	Consider an arbitrary node $w \in V - T$ that flooded $0$ (resp. $1$) in step (a) of this phase.
	Note that $w \in Z$ (resp. $w \in N$).
	Observe that $P_{wv}$ identified in step (b) of the phase excludes $T \cup F$ and is fault-free.
	Therefore, by Observation \ref{observation hybrid fault-free reliable}, $v$ receives $0$ (resp. $1$) along $P_{wv}$ and correctly puts $w$ in the set $Z_v$ (resp. $N_v$), as required.
	
	It follows that for any two non-faulty nodes $u$ and $v$, we have that $Z_u = Z_v = Z$ and $N_u = N_v = N$.
	Thus $A_u = A_v$ and $B_u = B_v$.
	Let $A := A_u$ and $B := B_u$.
	Since $T = T^*$ has no non-faulty node, a non-faulty node is either in $A$ or in $B$.
	First note that $A$ is always non-empty as follows.
	If $B$ is empty, then $A = V - T$ is non-empty since $n > f$.
	If $B$ is non-empty, then let $w \in B$ be a node in $B$.
	By Lemma \ref{lemma hybrid propagates}, there exist $f + 1$ node-disjoint $Aw$-paths, which implies that $\abs{A} \ge f + 1$.
	Now all nodes in $A$ flooded identical value in step (a), say $\alpha$.
	If $u \in A$, then $u$'s state is $\alpha$ at the beginning of the phase and stays unchanged in step (c).
	Therefore, at the end of the phase $\gamma_u = \alpha$.
	If $u \in B$, then observe that the $f+1$ node-disjoint $A u$-paths identified by $u$ in step (c) are all fault-free.
	By Observation \ref{observation hybrid fault-free reliable}, it follows that $u$ receives $\alpha$ identically along these $f+1$ paths and so, at the end of the phase, $\gamma_u = \alpha$.
	Similarly for $v$, we have that $\gamma_v = \alpha$, as required.
\end{proof}

We now have all the necessary ingredients to prove the sufficiency portion of Theorem \ref{theorem hybrid}.

\begin{proof_of}{Theorem \ref{theorem hybrid}}
    The necessity of the conditions in Theorem \ref{theorem hybrid} were proven in Appendix \ref{section proofs hybrid necessity}.
    For sufficiency, we prove the correctness of Algorithm \ref{algorithm hybrid consensus}.
    
    The algorithm terminates in finite time because the number of phases is finite and flooding in each phase completes in finite time.
	Thus, the algorithm satisfies the \emph{termination} condition.
	
	Since there are at most $f$ faulty nodes with at most $t$ equivocating faulty nodes in any given execution, there exists at least one phase in which sets $F = F^*$ and $T = T^*$.
	Then, from Lemma \ref{lemma hybrid agreement}, we have that all non-faulty nodes have the same state at the end of this phase.
	Lemma \ref{lemma hybrid validity} implies that the state of the non-faulty nodes will remain unchanged after any subsequent phases.
	Therefore, all non-faulty nodes will have the same state at the end of the algorithm, and their output will be identical.
	This proves that the algorithm satisfies the \emph{agreement} condition.
	
	At the start of phase 1, the state of each non-faulty node equals its own input.
	Now, applying Lemma \ref{lemma hybrid validity} inductively implies that the state of a non-faulty node always equals the {\em input} of some non-faulty node.
	This, in turn, implies that the algorithm satisfies the {\em validity} condition.
	Thus, we have proved correctness of Algorithm \ref{algorithm hybrid consensus} under the conditions stated in Theorem \ref{theorem hybrid}.
\end{proof_of}